\documentclass[a4paper,fleqn,11pt]{article}

\usepackage{amsmath}
\usepackage{amsthm}
\usepackage{amssymb}
\usepackage{mathtools}
\usepackage{accents}
\usepackage{float}
\usepackage{comment}

\usepackage[a4paper,top=3cm, bottom=3cm, left=3cm, right=3cm]{geometry}
\usepackage[shortlabels,inline]{enumitem}
\usepackage[pdftex]{graphicx}
\usepackage[pdftex,ocgcolorlinks,pagebackref=false]{hyperref}
\usepackage{authblk}

\setlist[enumerate,1]{label={(\roman*)}}

\usepackage[capitalise,noabbrev]{cleveref}

\Crefname{property}{Property}{Properties}

\theoremstyle{plain}
\newtheorem{theorem}{Theorem}[section]
\newtheorem{lemma}[theorem]{Lemma}
\newtheorem{proposition}[theorem]{Proposition}
\newtheorem{corollary}[theorem]{Corollary}
\newtheorem{remark}[theorem]{Remark}
\newtheorem{example}[theorem]{Example}

\theoremstyle{definition}
\newtheorem{definition}[theorem]{Definition}

\usepackage{pgfplots}
\pgfplotsset{compat=1.16}
\usepgfplotslibrary{fillbetween}
\usetikzlibrary{patterns}
\pgfplotsset{
  spectrumplotstyle/.style={
    grid=both, 
    axis lines=middle,
    xmin=0,
    xmax=1.1,
    xtick={1},
    extra x ticks={0},
    xlabel=$\alpha$,
    xlabel style={anchor=north},
    ylabel style={anchor=south east},
    width=0.6\textwidth,
    height=0.45\textwidth,
  },
  rateplotstyle/.style={
    grid=both, 
    axis lines=middle,
    xmin=0,
    xmax=0.7,
    extra x ticks={0},
    xlabel=$r$,
    xlabel style={anchor=north},
    ylabel=$R$,
    ylabel style={anchor=south east},
    width=0.6\textwidth, 
    height=0.45\textwidth, 
  },
  plotline/.style={
    black,
    semithick,
    mark=none,
  },
  plotlinewithmarks/.style={
    black,
    semithick,
    mark=*,
    mark size=1 pt,
  },
  exampleline/.style={
    black,
    semithick,
    dashed,
    mark=none,
  },
  partofspectrumregion/.style={
    pattern=north west lines,
    pattern color=black!50!white
  },
  notpartofspectrumregion/.style={
    black!50!white,
    fill,
    opacity=0.5
  },
  achievableregion/.style={
    black!50!white,
    fill,
    opacity=0.5
  },
  notachievableregion/.style={
    pattern=north west lines,
    pattern color=black!50!white
  },
}
\tikzset{
  examplepoint/.style={
    black,
    fill,
    circle,
    scale=0.3,
  },
}

\pgfplotstableread[col sep=comma]{
R,r
0.0       , 1.0819573016638942e-12
0.025     , 1.3988283111943112e-05
0.05      , 0.0001085567381645669
0.07500000000000001, 0.00035956859740010616
0.1       , 0.0008346539564101363
0.125     , 0.0016085047236342742
0.15000000000000002, 0.0027397266225869243
0.17500000000000002, 0.004275768189205387
0.2       , 0.006268256901556524
0.225     , 0.008756796777012466
0.25      , 0.011769328318738892
0.275     , 0.015347084329016658
0.30000000000000004, 0.019479981428590865
0.325     , 0.02423134752799294 
0.35000000000000003, 0.029584707923977493
0.375     , 0.03558545315759987 
0.4       , 0.04221568058779179 
0.42500000000000004, 0.04949645211632814 
0.45      , 0.057469955632002634
0.47500000000000003, 0.06607822820841314 
0.5       , 0.07536991727688602 
0.525     , 0.08535487076477533 
0.55      , 0.0960400469558087  
0.5750000000000001, 0.10744985673330731 
0.6000000000000001, 0.11949085001675686 
0.625     , 0.1322843050280117  
0.65      , 0.14580007852743848 
0.675     , 0.1600329775598347  
0.7000000000000001, 0.1749742012158063  
0.7250000000000001, 0.19073240985925366 
0.75      , 0.20715845398430643 
0.775     , 0.22432454718836492 
0.8       , 0.24229581189968008 
0.8250000000000001, 0.26107747449749885 
0.8500000000000001, 0.2806803017863535  
0.875     , 0.30107633954892676 
0.9       , 0.32217264900652576 
0.925     , 0.34410441268793035 
0.9500000000000001, 0.366884090041545   
0.9750000000000001, 0.39052606033996906 
1.0       , 0.4153840892127035  
}\invertedfunctiondata

\newcommand{\setbuild}[2]{\left\{#1\middle|#2\right\}}

\newcommand{\ed}{\mathop{}\!\mathrm{d}}
\newcommand{\norm}[2][]{\left\|#2\right\|_{#1}}

\newcommand{\ket}[1]{\left|#1\right\rangle}
\newcommand{\bra}[1]{\left\langle #1\right|}
\newcommand{\ketbra}[2]{\left|#1\middle\rangle\!\middle\langle#2\right|}
\newcommand{\braket}[2]{\left\langle#1\middle|#2\right\rangle}
\newcommand{\vectorstate}[1]{\ketbra{#1}{#1}}

\newcommand{\EPR}{\textnormal{EPR}}
\newcommand{\GHZ}{\textnormal{GHZ}}
\newcommand{\W}{\textnormal{W}}

\newcommand{\ball}[2]{B_{#1}(#2)}
\newcommand{\closedball}[2]{\overline{B}_{#1}(#2)}
\DeclareMathOperator{\boundeds}{\mathcal{B}}

\newcommand{\entropy}{H}

\newcommand{\relativeentropy}[3][]{\mathop{D_{#1}}\mathopen{}\left(#2\middle\|#3\right)\mathclose{}}

\newcommand{\floor}[1]{\lfloor #1\rfloor}

\newcommand{\distributions}[1][]{\mathcal{P}_{#1}}

\newcommand{\reals}{\mathbb{R}}
\newcommand{\complexes}{\mathbb{C}}
\newcommand{\naturals}{\mathbb{N}}
\newcommand{\integers}{\mathbb{Z}}
\newcommand{\nonnegativereals}{\mathbb{R}_{\ge 0}}
\newcommand{\positivereals}{\mathbb{R}_{>0}}

\newcommand{\unitcircle}{\mathbb{T}}

\newcommand{\positiveintegers}{\mathbb{N}_{>0}}

\newcommand{\distance}{d}

\let\Re\undefined
\let\Im\undefined
\DeclareMathOperator{\Re}{Re}
\DeclareMathOperator{\Im}{Im}

\DeclareMathOperator{\support}{supp}

\DeclareMathOperator{\supp}{supp}

\DeclareMathOperator{\Hom}{Hom}

\title{Error exponents for entanglement transformations from degenerations}

\author[1]{D\'avid Bug\'ar}
\author[1]{P\'eter Vrana}
\affil[1]{Department of Algebra and Geometry, Institute of Mathematics, Budapest University of Technology and Economics, M\H uegyetem~rkp. 3., H-1111 Budapest, Hungary.}
\date{}

\begin{document}

\maketitle

\begin{abstract}
This paper explores the trade-off relation between the rate and the strong converse exponent for asymptotic LOCC transformations between pure multipartite states. Any single-copy probabilistic transformation between a pair of states implies that an asymptotic transformation at rate $1$ is possible with an exponentially decreasing success probability. However, it is possible that an asymptotic transformation is feasible with nonzero probability, but there is no transformation between any finite number of copies with the same rate, even probabilistically. In such cases it is not known if the optimal success probability decreases exponentially or faster. A fundamental tool for showing the feasibility of an asymptotic transformation is degeneration. Any degeneration gives rise to a sequence of stochastic LOCC transformations from copies of the initial state plus a sublinear number of GHZ states to the same number of copies of the target state. These protocols involve parameters that can be freely chosen, but the choice affects the success probability. In this paper, we characterize an asymptotically optimal choice of the parameters and derive a single-letter expression for the error exponent of the resulting protocol. In particular, this implies an exponential lower bound on the success probability when the stochastic transformation arises from a degeneration.
\end{abstract}

\section{Introduction}

Local operations and classical communication (LOCC) involve multiple parties who manipulate their respective subsystems using local quantum operations and coordinate these actions through classical communication channels \cite{bennett1996mixed}. This operational framework is central to understanding, characterizing, and quantifying entanglement. Entanglement distillation \cite{bennett1996concentrating}, entanglement swapping, and quantum teleportation \cite{bennett1993teleporting} are prime examples of protocols that rely on LOCC operations and utilize entanglement. 

In an asymptotic setting with given one copy states $\rho$ and $\sigma$ we aim to study transformations $\rho^{\otimes n}\to\sigma^{\otimes m}$ while $n$ and $m$ grow. The main object of interest here is the transformation rate $R=\frac{m}{n}$, which encapsulates the efficiency of the transformation. It reflects the underlying entanglement structure and the operational feasibility of converting one state into another within the LOCC framework. Given the difficulty of finding the best achievable transformation rates, this problem has several relaxations. Instead of exact transformations one can allow the success probability of the transformation to be less than one, or allow the output state to differ from the target. When relaxing on the probability one can require exponentially decaying failure probability (direct regime), exponentially decaying success probability (converse regime) or even arbitrary nonzero probability (asymptotic SLOCC paradigm). For the special case of bipartite entanglement concentration, i.e., transforming many copies of an arbitrary bipartite pure state into $\EPR$ pairs, the optimal rate is known in the full range between asymptotic SLOCC and deterministic transformations \cite{bennett1996concentrating,morikoshi2001deterministic,hayashi2002error}.

In the SLOCC paradigm we ask if the pure state $\psi$ can be transformed into the pure state $\varphi$ with arbitrarily small, but nonzero probability. It was found in \cite{bennett2000exact,dur2000three} that this is equivalent to asking if there exist linear maps $A_j$ such that $(A_1\otimes\dots\otimes A_k)\psi=\varphi$. In that case, we say that the state $\psi$ can be transformed by SLOCC into $\varphi$ or, in the language of tensors, the tensor $\psi$ restricts to $\varphi$. 
In \cite{chitambar2008tripartite} Chitambar, Duan, and Shi noted that the asymptotic restrictions over complex numbers can be seen as asymptotic transformations between pure multipartite states with SLOCC. This is defined as follows: we say that $\psi$ asymptotically restricts to $\varphi$ if there exists a sequence of natural numbers $(r_n)_{n\in\naturals}$ such that $\sqrt[n]{r_n}\to 1$ and $\GHZ_{r_n} \otimes \psi^{\otimes n}$ restricts to $\varphi^{\otimes n}$, where $\GHZ_{r_n}$ denotes the $r_n$-level $\GHZ$ state. Asymptotic entanglement transformations from this viewpoint have been studied in several papers \cite{yu2010tensor,chen2010tensor,yu2014obtaining,vrana2015asymptotic,vrana2017entanglement,christandl2016asymptotic,christandl2023universal,gharahi2024persistent}.

It should be noted that asymptotic SLOCC transformations are a rather weak notion of entanglement transformations, in principle allowing the success probability to approach $0$ arbitrarily fast. As a refinement of the previous setting, in the converse regime we allow transformations $\psi^{\otimes n}\to\varphi^{\otimes Rn+o(n)}$ with exponentially decaying probability $2^{-rn+o(n)}$ for a specified \emph{converse error exponent} $r$. In the bipartite case, the precise trade-off relation between the rate $R$ and the error exponent $r$ is known \cite{hayashi2002error,jensen2019asymptotic}. For more than two parties, the problem of explicitly describing the achievable pairs $(R,r)$ is wide open (see \cite{jensen2019asymptotic,vrana2023family,bugar2022interpolating,bugar2024explicit} for partial results). By finding specific protocols, one can show the achievability of certain pairs $(R,r)$ for given pure states $\psi$ and $\varphi$. In the simplest case, suppose that a single-copy SLOCC transformation exists from $\psi$ to $\varphi$. If the transformation is successful with probability $p$ then, by running the protocol independently on many copies, we obtain an asymptotic transformation with rate $1$ and error exponent $-\log p$. More generally, by accepting an outcome with at least $Rn$ successful runs for some $R\in(0,1)$, it follows from a standard tail bound for the binomial distribution (see, e.g., \cite[Problem 2.8]{csiszar2011information}) that the error exponent $d(R\|p)$ is achievable, where $d(q\|p)=q\log\frac{q}{p}+(1-q)\log\frac{1-q}{1-p}$ is the Kullback--Leibler divergence between two Bernoulli distributions.

Interestingly, when the number of parties is at least $3$, there exist states $\psi$ and $\varphi$ such that $\psi$ cannot be transformed into $\varphi$ by SLOCC, but $\varphi$ can still be approximated arbitrarily closely with states in the SLOCC orbit of $\psi$. This was recognized in the quantum information literature in \cite{walther2005local} for transformations of $\GHZ$ states into approximate $\W$ states, and previously in the context of algebraic computations \cite{bini1979complexity,bini1980relations,schonhage1981partial}, where it is known as a tensor degeneration.

The significance of degenerations is that, like single- or multicopy SLOCC transformations, they imply asymptotic SLOCC transformations in the sense that if $\psi$ degenerates to $\varphi$, then $\psi^{\otimes n}$ together with a sublinear number of $\GHZ$ states can be transformed into $\varphi^{\otimes n}$ by SLOCC \cite{bini1980relations}. However, it is not immediately clear if the transformation is possible with a finite error exponent. In this paper, we show that degeneration does imply a finite strong converse exponent, and find a single-letter upper bound on the exponent in terms of the data specifying the degeneration.

Our starting point is the algebraic definition of degeneration (see \cite{alder1984grenzrang,strassen1987relative} or \cite[(20.24) Theorem]{burgisser2013algebraic} for the equivalence with the aforementioned topological one): we say that $\psi$ degenerates to $\varphi$ if there exists local linear map-valued complex Laurent polynomials $A_j(z)$ such that $(A_1(z)\otimes\dots\otimes A_k(z))\psi=\varphi +O(z)$ (see \cref{sec:DegenLOCC} for more details). Our main result is the following:
\begin{theorem}\label{thm:main}
Let $\psi$ and $\varphi$ be $k$-partite state vectors (normalized as $\norm{\psi}=\norm{\varphi}=1$), and let $(A_1(z)\otimes\dots\otimes A_k(z))\psi=\varphi +O(z)$ be a degeneration, where the right hand side is a polynomial of degree $e$. Suppose that $A_1(z)\otimes\dots\otimes A_k(z)$ is nowhere zero and contains terms with positive as well as negative powers of $z$. Then $\psi$ can be asymptotically transformed into $\varphi$ with rate $R=1$ and strong converse exponent at most
\begin{equation}
2 \inf_{\sigma\in\distributions(\complexes)}\sup_{z\in \supp \sigma} \left[
        \log \norm{A_1(z)\otimes\dots\otimes A_k(z)}+
        e\int_\complexes
        \log\frac{\lvert t\rvert}{\lvert z-t\rvert}\ed\sigma(t)\right],
\end{equation}
where $\distributions(\complexes)$ denotes the set of probability distributions on the complex plane.
\end{theorem}
We note that the conditions on $A_1(z)\otimes\dots\otimes A_k(z)$ are not too restrictive, as any degeneration that does not satisfy these can either be turned into a restriction or a degeneration satisfying the conditions (\cref{rem:degenInterestingcases}).

The outline of the paper is the following. In \cref{sec:DegenLOCC} we present a family of LOCC protocols based on the standard proof showing that degenerations can be turned into SLOCC transformations. The family depends on the choice of a set of complex numbers for each number of copies, which one needs to optimize in order to obtain the best possible error exponent. In \cref{sec:limitofoptprob} we give an expression for the optimal error exponent, which depends only on the norm of the linear map valued Laurent polynomials present in the degeneration, and the approximation degree $e$. This corresponds to showing achievable error exponents for the transformation rate $R=1$. In \cref{sec:bounds} we derive upper and lower bounds on the optimized expression, and discuss in detail the special case when $\norm{A_1(z)\otimes\dots\otimes A_k(z)}$ is centrally symmetric. In \cref{sec:errorexptrfratelessone} we present a generalization of the protocol that allows us to trade the rate for the success probability, leading to smaller achievable strong converse exponents for $R<1$.

\section{LOCC transformations from degeneration}\label{sec:DegenLOCC}

We recall the algebraic definition of a degeneration (see, e.g., \cite[(15.19) Definition]{burgisser2013algebraic}).
\begin{definition}
    Let $\psi\in\mathcal{H}_1\otimes\dots\otimes\mathcal{H}_k$ and $\varphi\in\mathcal{K}_1\otimes\dots\otimes\mathcal{K}_k$ be unit vectors and $A_j\in\complexes[z,z^{-1}]\otimes\Hom(\mathcal{H}_j,\mathcal{K}_j)$ be linear map-valued Laurent polynomials ($j=1,\dots,k$) such that
\begin{equation}\label{eq:degenpolynomial}
(A_1(z)\otimes\dots\otimes A_k(z))\psi=\sum_{h=0}^e z^h\varphi_h=\varphi+O(z),
\end{equation}
where the right hand side is a polynomial in $z$ with constant term $\varphi_0=\varphi$.
We call this a \emph{degeneration} between $\psi$ and $\varphi$.
The degree $e$ of the polynomial is called the error degree. 
\end{definition}

\begin{remark}\label{rem:degenInterestingcases}
The only interesting case of degeneration is when $A_1(z)\otimes\dots\otimes A_k(z)$ contains both positive and negative powers of $z$. If it contains no negative powers, this relation reduces to a \emph{restriction}. On the other hand, if there are no positive powers in the Laurent polynomial, then the negative powers only contribute to the negative powers on the right hand side (which vanishes by assumption), therefore these can be omitted as well, which means that the states are also related by a restriction.
 
Sometimes we also need the assumption that $A_1(z)\otimes\dots\otimes A_k(z)$ is nowhere zero. Even if we start with a degeneration for which this does not hold, we can derive another degeneration with this property. If for some $z_0$ we have $A_1(z_0)\otimes\dots\otimes A_k(z_0)=0$, then for some $j$ every matrix element of $A_j(z)$ is divisible by $z-z_0$, so dividing it by $z-z_0$ leads to a degeneration with smaller-degree polynomials. Repeating this procedure for all zeros leads to a degeneration satisfying the assumption.
\end{remark}

By taking the $n$-th tensor power of \eqref{eq:degenpolynomial} we get
\begin{equation}
(A_1(z)\otimes\dots\otimes A_k(z))^{\otimes n}\psi^{\otimes n}=\varphi^{\otimes n}+O(z),    
\end{equation}
where the error degree, i.e., the the degree of the polynomial on the right hand side is $ne$.

Our aim is to turn this into a probabilistic LOCC transformation from a (low-rank, weighted) GHZ state times $\psi$ to $\varphi$ with a large success probability.

To this end, let $t\in\naturals$ such that $t\ge ne+1$, and consider complex numbers $z_i$, $c_i$ and $\gamma_{j,i}\neq 0$ with $j=1,\dots,k$ and $i=1,\dots,t$ subject to
\begin{equation}\label{eq:czconditions}
\underbrace{\begin{bmatrix}
1 & 1 & \dots & 1  \\
z_1 & z_2 & \dots & z_t  \\
\vdots & \vdots & \ddots & \vdots  \\
z_1^e & z_2^e & \dots & z_t^e
\end{bmatrix}}_Z
\cdot\underbrace{\begin{bmatrix}
c_1  \\
c_2  \\
\vdots  \\
c_t
\end{bmatrix}}_c
=\underbrace{\begin{bmatrix}
1  \\
0  \\
\vdots  \\
0
\end{bmatrix}}_{e_1}
\end{equation}
and
\begin{equation}\label{eq:gammasumbound}
\forall j:\norm{\sum_{i=1}^t\lvert\gamma_{j,i}\rvert^2 A(z_i)^{\otimes n}A^*(z_i)^{\otimes n}}\le 1.
\end{equation}
Then the maps
\begin{equation}
\tilde{A}_j:=\sum_{i=1}^t\gamma_{j,i}\bra{i}\otimes A_j(z_i)^{\otimes n}\in\Hom(\complexes^t\otimes\mathcal{H}_j,\mathcal{K}_j)
\end{equation}
are contractions and with $g_i=\prod_{j=1}^k\gamma_{j,i}$ we have
\begin{equation}\label{eq:procedureTophi}
\begin{split}
(\tilde{A}_1\otimes\dots\otimes\tilde{A}_k)\left(\sum_{i=1}^t \frac{c_i}{g_i}\ket{i}\otimes\psi^{\otimes n}\right)
 & = \sum_{i=1}^t \frac{c_i}{g_i}((\gamma_{1,i}A_1(z_i)^{\otimes n})\otimes\dots\otimes (\gamma_{k,i}A_k(z_i)^{\otimes n}))\psi^{\otimes n}  \\
 & = \sum_{i=1}^t c_i A_1(z_i)^{\otimes n}\otimes\dots\otimes A_k(z_i)^{\otimes n}\psi^{\otimes n}  \\
 & = \sum_{i=1}^t c_i \sum_{h=0}^{ne} z_i^h\varphi_h^{\otimes n}  \\
 & = \sum_{h=0}^{ne}\left(\sum_{i=1}^t z_i^hc_i\right)\varphi_h^{\otimes n}  \\
 & = \varphi^{\otimes n},
\end{split}
\end{equation}
which implies that a rank-$t$ uniform $\GHZ$ state and $\psi^{\otimes n}$ can be transformed into $\varphi^{\otimes n}$ with probability at least
\begin{equation}\label{eq:problowerbound}
p\left(\GHZ_t\otimes\psi^{\otimes n}\to\varphi^{\otimes n}\right) \ge \left(\sum_{i=1}^t \frac{\lvert c_i\rvert^2}{\lvert g_i\rvert^2}\right)^{-1},
\end{equation}
where the right hand side implicitly depends on $z_1,\dots,z_t$ and $\gamma_{1,1},\dots\gamma_{k,t}$ through the conditions \eqref{eq:czconditions} and \eqref{eq:gammasumbound}.

\begin{example}\label{ex:rootsofunity}
For $t\ge e+1$ we can choose $z_i=\omega^{i-1}$ where $\omega$ is a primitive $t$-th root of unity, $c_1=c_2=\dots=c_t=\frac{1}{t}$, $\gamma_{j,i}=\frac{1}{\sqrt{t}\norm{A(z_i)}^n}$. This gives the lower bound
\begin{equation}
p\left(\GHZ_t\otimes\psi^{\otimes n}\to\varphi^{\otimes n}\right)\ge t^{1-k}\frac{1}{\max_i\prod_{j=1}^k\norm{A_j(z_i)}^{2n}}\ge t^{1-k}\frac{1}{\max_{z\in\unitcircle}\prod_{j=1}^k\norm{A_j(z)}^{2n}},
\end{equation}
where $\unitcircle=\setbuild{z\in\complexes}{\lvert z\rvert=1}$.
\end{example}

\subsection{Optimizing the coefficients}

The protocol described above depends on several choices in addition to the degeneration itself: the number and the position of the interpolation points $z_1,\dots,z_t$, the coefficients $c_1,\dots,c_t$ of the (unnormalized GHZ state), and the factors $\gamma_{j,i}$, subject to \eqref{eq:czconditions} and \eqref{eq:gammasumbound}. Our goal is to maximize the resulting lower bound \eqref{eq:problowerbound} on the probability as a function of the number $n$ of copies, up to subexponential factors. In this section, considering fixed $z_1,\dots,z_t$, we first find exactly the optimal coefficients $c_1,\dots,c_{t}$, then an approximately optimal choice of $\gamma_{j,i}$, and use these to show that we may assume $t=ne+1$ (i.e., the minimal value), finally arriving at an expression that has the same exponential behaviour as the optimal probability, and that depends only on the $ne+1$ points $z_0,\dots,z_{ne}\in\complexes\setminus\{0\}$. For technical reasons it will be convenient to restrict the optimization to a compact subset $K\subseteq\complexes\setminus\{0\}$ (independent of $n$). We show that this can be done with a small loss (that vanishes as $K\to\complexes$) in the error exponent. In the following we abbreviate the transformation probability by $p$.

\begin{proposition}\label{prop:optbyc}
    Let us fix $t$, $z_1,\dots,z_t$ distinct and $\gamma_{1,1},\ldots,\gamma_{k,t}$. Then by optimizing the right hand side of \eqref{eq:problowerbound}
over the $c_i$ factors we get
\begin{equation}\label{eq:pAz}
p\ge
\bra{e_1}
\left(\sum_{i=1}^t\lvert g_i\rvert^2\begin{bmatrix}
1  \\
z_i  \\
\vdots  \\
z_i^{ne}
\end{bmatrix}\cdot\begin{bmatrix}
1 & \overline{z_i} & \cdots & \overline{z_i}^{ne}
\end{bmatrix}
\right)^{-1}
\ket{e_1}
\end{equation}
which is the upper left corner of the inverse of an $(ne+1)\times(ne+1)$ matrix, which is the corresponding minor divided by the determinant.

\end{proposition}

\begin{proof}
First we introduce the diagonal matrix $G$ with diagonal entries $\lvert g_i\rvert^{2}$. The task is to minimize $\langle c,G^{-1}c\rangle$ subject to the condition \eqref{eq:czconditions} written as $Zc=e_1$. Let $x=\sqrt{G^{-1}}c$ so that the objective function is $\langle c,G^{-1}c\rangle=\norm{x}^2$ and the condition is $ZG^{1/2}x=e_1$, which is equivalent to
\begin{equation}
G^{1/2}Z^*ZG^{1/2}x=G^{1/2}Z^*e_1
\end{equation}
since $G^{1/2}Z^*$ is injective. If $x=x_\perp+x_\parallel$ with $x_\parallel\in\supp G^{1/2}Z^*ZG^{1/2}$ and $x_\perp\in\ker G^{1/2}Z^*ZG^{1/2}$, then $G^{1/2}Z^*ZG^{1/2}x=G^{1/2}Z^*ZG^{1/2}x_\parallel$ and $\norm{x}^2=\norm{x_\perp}^2+\norm{x_\parallel}^2\ge\norm{x_\parallel}^2$, therefore the optimal $x$ satisfies $x\in\supp G^{1/2}Z^*ZG^{1/2}$. On this subspace $G^{1/2}Z^*ZG^{1/2}$ is injective, therefore
\begin{equation}
x=(G^{1/2}Z^*ZG^{1/2})^+G^{1/2}Z^*e_1,
\end{equation}
where $(\cdot)^+$ is the Moore--Penrose inverse.

Since $G^{1/2}Z^*\in\complexes^{t\times(ne+1)}$ and $ZG^{1/2}\in\complexes^{(ne+1)\times t}$ both have rank $ne+1$, i.e., $(G^{1/2}Z^*)\cdot(ZG^{1/2})$ is a full-rank factorization and we can use \cite[Fact 6.4.9.]{bernstein2009matrix}\footnote{If $A\in\complexes^{n\times r}$ and $B\in\complexes^{r\times m}$ both have rank $r$, then $(AB)^+=B^*(BB^*)^{-1}(A^*A)^{-1}A^*$} to find the inverse:
\begin{equation}
(G^{1/2}Z^*ZG^{1/2})^+=G^{1/2}Z^*(ZGZ^*)^{-2}ZG^{1/2},
\end{equation}
therefore
\begin{equation}
\begin{split}
x
 & = G^{1/2}Z^*(ZGZ^*)^{-2}ZG^{1/2}G^{1/2}Z^*e_1  \\
 & = G^{1/2}Z^*(ZGZ^*)^{-1}e_1.
\end{split}
\end{equation}
This gives for the optimal $c$
\begin{equation}\label{eq:optcGc}
\begin{split}
\langle c,Gc\rangle
 & = \norm{x}^2  \\
 & = \langle G^{1/2}Z^*(ZGZ^*)^{-1}e_1,G^{1/2}Z^*(ZGZ^*)^{-1}e_1\rangle  \\
 & = \langle e_1,(ZGZ^*)^{-1}e_1\rangle.
\end{split}
\end{equation}
Note that this is the upper left corner of the inverse of an $(ne+1)\times(ne+1)$ matrix, which is the corresponding minor divided by the determinant. That matrix can be written as
\begin{equation}\label{eq:ZGZ}
ZGZ^*=\sum_{i=1}^t\lvert g_i\rvert^2\begin{bmatrix}
1  \\
z_i  \\
\vdots  \\
z_i^{ne}
\end{bmatrix}\cdot\begin{bmatrix}
1 & \overline{z_i} & \cdots & \overline{z_i}^{ne}
\end{bmatrix}
\end{equation}
\end{proof}

\begin{lemma}\label{lem:optzgz}
The optimal $G$ in \eqref{eq:ZGZ} (matching up to a subexponential factor) can be written
    \begin{equation}\label{eq:optzgz}
ZGZ^*
 = \sum_{i=1}^t t^{-k}\frac{1}{\norm{A(z_i)}^{2n}}\begin{bmatrix}
1  \\
z_i  \\
\vdots  \\
z_i^{ne}
\end{bmatrix}\cdot\begin{bmatrix}
1 & \overline{z_i} & \cdots & \overline{z_i}^{ne}.
\end{bmatrix}
\end{equation}
\end{lemma}

\begin{proof}
If we consider the numbers $z_i$ and $c_i$ fixed and increase $\lvert\gamma_{j,i}\rvert$ then the lower bound increases, therefore the optimal choice saturates \eqref{eq:gammasumbound}. We have the inequalities
\begin{equation*}
\begin{split}
\max_{i}\lvert\gamma_{j,i}\rvert^2\norm{A_j(z_i)}^{2n}
 & \le \norm{\sum_{i=1}^t\lvert\gamma_{j,i}\rvert^2 A_j(z_i)^{\otimes n}A^*_j(z_i)^{\otimes n}}  \\
 & \le \sum_{i=1}^t\norm{\lvert\gamma_{j,i}\rvert^2 A_j(z_i)^{\otimes n}A^*_j(z_i)^{\otimes n}}  \\
 & \le t\max_{i}\lvert\gamma_{j,i}\rvert^2\norm{A_j(z_i)}^{2n},
\end{split}
\end{equation*}
therefore for every $i,j$ we have $\lvert\gamma_{j,i}\rvert^2\norm{A_j(z_i)}^{2n}\le 1$ and for all $j$ there is an $i$ such that
\begin{equation}
\lvert\gamma_{j,i}\rvert^2\norm{A_j(z_i)}^{2n}\ge\frac{1}{t}.
\end{equation}

\begin{equation}
\lvert g_i\rvert^2
 = \prod_{j=1}^k\lvert\gamma_{j,i}\rvert^2
 \le \prod_{j=1}^k\frac{1}{\norm{A_j(z_i)}^{2n}}
 = \frac{1}{\norm{A(z_i)}^{2n}},
\end{equation}
where $A(z)=A_1(z)\otimes\dots\otimes A_k(z)$.

On the other hand, it is possible to choose for all $i,j$: $\lvert\gamma_{j,i}\rvert^2=\frac{1}{t}\frac{1}{\norm{A_j(z_i)}^{2n}}$, this gives
\begin{equation}
\lvert g_i\rvert^2
 = \prod_{j=1}^k\lvert\gamma_{j,i}\rvert^2
 = \prod_{j=1}^k \frac{1}{t}\frac{1}{\norm{A_j(z_i)}^{2n}}
 = t^{-k}\frac{1}{\norm{A(z_i)}^{2n}},
\end{equation}
i.e., we have upper and lower bounds on the optimal $G$, matching up to a subexponential factor. These give
\begin{equation}
\begin{split}
ZGZ^*
 & = \sum_{i=1}^t\lvert g_i\rvert^2\begin{bmatrix}
1  \\
z_i  \\
\vdots  \\
z_i^{ne}
\end{bmatrix}\cdot\begin{bmatrix}
1 & \overline{z_i} & \cdots & \overline{z_i}^{ne}
\end{bmatrix}  \\
 & = \sum_{i=1}^t t^{-k}\frac{1}{\norm{A(z_i)}^{2n}}\begin{bmatrix}
1  \\
z_i  \\
\vdots  \\
z_i^{ne}
\end{bmatrix}\cdot\begin{bmatrix}
1 & \overline{z_i} & \cdots & \overline{z_i}^{ne}
\end{bmatrix}
\end{split}
\end{equation}
\end{proof}

Note that if we keep only $ne+1$ terms in the sum, then the matrix decreases, therefore its inverse increases. In the following we show that by doing so we do not lose anything in terms of the error exponent.

\begin{lemma}\label{lem:optbyt}
Let $d,t\in\naturals$, $1\le d\le t$, $\ket{v_1},\ket{v_2},\dots,\ket{v_t}\in\complexes^d$ vectors in general position, and $\ket{u}\in\complexes^d$. Then there is a subset $T\in\binom{[t]}{d}$ such that
\begin{equation}
\bra{u}\left(\sum_{i\in T}\ketbra{v_i}{v_i}\right)^{-1}\ket{u}\le (t-d+1)\bra{u}\left(\sum_{i=1}^t\ketbra{v_i}{v_i}\right)^{-1}\ket{u}
\end{equation}
\end{lemma}
\begin{proof}
We prove by induction on $t$. If $t=d$, the two sides are equal, therefore the inequality is true. Otherwise let
\begin{equation}
\ket{w_j}=\left(\sum_{i=1}^t\ketbra{v_i}{v_i}\right)^{-\frac{1}{2}}\ket{v_j}
\end{equation}
and
\begin{equation}
\ket{f}=\left(\sum_{i=1}^t\ketbra{v_i}{v_i}\right)^{-\frac{1}{2}}\ket{u}.
\end{equation}
For each $j\in[t]$ we have
\begin{equation}
\begin{split}
\frac{\displaystyle\bra{u}\left(\sum_{i\in[t]\setminus\{j\}}\ketbra{v_i}{v_i}\right)^{-1}\ket{u}}{\displaystyle\bra{u}\left(\sum_{i=1}^t\ketbra{v_i}{v_i}\right)^{-1}\ket{u}}
 & = \frac{\displaystyle\bra{f}\left(\sum_{i\in[t]\setminus\{j\}}\ketbra{w_i}{w_i}\right)^{-1}\ket{f}}{\displaystyle\bra{f}\left(\sum_{i=1}^t\ketbra{w_i}{w_i}\right)^{-1}\ket{f}}  \\
 & = \frac{\displaystyle\bra{f}\left(I-\ketbra{w_j}{w_j}\right)^{-1}\ket{f}}{\displaystyle\braket{f}{f}}  \\
 & = \frac{\displaystyle\bra{f}\left(I+\frac{1}{1-\norm{w_j}^2}\ketbra{w_j}{w_j}\right)\ket{f}}{\displaystyle\braket{f}{f}}  \\
 & = 1+\frac{\braket{f}{w_j}\braket{w_j}{f}}{\norm{f}^2(1-\norm{w_j}^2)}.
\end{split}
\end{equation}
Note that $0\le\norm{w_j}<1$ and
\begin{equation}
\sum_{j=1}^t(1-\norm{w_j}^2)=t-d,
\end{equation}
therefore we can form the convex combination
\begin{equation}
\sum_{j=1}^t\frac{1-\norm{w_j}^2}{t-d}\frac{\braket{f}{w_j}\braket{w_j}{f}}{\norm{f}^2(1-\norm{w_j}^2)}=\frac{1}{t-d}\frac{1}{\norm{f}^2}\bra{f}\sum_{j=1}^t\ketbra{w_j}{w_j}\ket{f}=\frac{1}{t-d}.
\end{equation}
It follows that there is an index $j$ such that
\begin{equation}
1+\frac{\braket{f}{w_j}\braket{w_j}{f}}{\norm{f}^2(1-\norm{w_j}^2)}\le 1+\frac{1}{t-d}=\frac{t-d+1}{t-d}
\end{equation}

Choosing a $j$ with this property, the set of $t-1$ vectors $\ket{v_1},\dots,\ket{v_{j-1}},\ket{v_{j+1}},\dots,\ket{v_t}$ are in general position, therefore, by the induction hypothesis, there is a subset $T\subseteq[t]\setminus\{j\}$ such that
\begin{equation}
\begin{split}
\bra{e}\left(\sum_{i\in T}\ketbra{v_i}{v_i}\right)^{-1}\ket{e}
 & \le (t-d)\bra{e}\left(\sum_{i\in[t]\setminus\{j\}}\ketbra{v_i}{v_i}\right)^{-1}\ket{e}  \\
 & \le (t-d)\frac{t-d+1}{t-d}\bra{e}\left(\sum_{i=1}^t\ketbra{v_i}{v_i}\right)^{-1}\ket{e}  \\
 & \le (t-d+1)\bra{e}\left(\sum_{i=1}^t\ketbra{v_i}{v_i}\right)^{-1}\ket{e}.
\end{split}
\end{equation}
\end{proof}

This proposition implies that it suffices to consider $t=ne+1$ terms (i.e., the minimal number) for the $n$th tensor power, as any subexponential sequence $t_n$ would only increase the success probability by a subexponential factor (at most $\frac{t_n}{ne+1}$).

\begin{definition}\label{def:errorexp}
We say that $r\in\nonnegativereals$ is an achievable exponent if 
\begin{equation}
r=-\liminf_{n\to\infty}\frac{1}{n}\log p_n(A_1(z)\otimes\dots\otimes A_k(z)),
\end{equation}
where $p_n$ is the probability of the transformation depending on the number of copies $n$ of the input state and on $\norm{A_1(z)\otimes\dots\otimes A_k(z)}=\prod_{j=1}^k\norm{A_j(z)}$.
\end{definition}

\begin{example}
\cref{ex:rootsofunity} with $t_n=ne+1$ implies that
\begin{equation}
r=\max_{z\in \unitcircle}\sum_{j=1}^k\log\norm{A_j(z)}^{2n}
\end{equation}
is an achievable exponent where $\unitcircle$ is the unit circle.
\end{example}

\begin{corollary}\label{cor:finalcountable}
Up to a polynomial factor the optimal probability in \eqref{eq:problowerbound} is
\begin{equation}\label{eq:optcGcmaineq}
    \langle c,G^{-1}c\rangle= \sum_{i=1}^{ne+1} \prod_{l\neq i} \norm{A_1(z_i)\otimes\dots\otimes A_k(z_i)}^{2n}\frac{\lvert z_l\rvert^2}{\lvert z_i-z_l\rvert^2}.
\end{equation}

\end{corollary}
\begin{proof}
    By \cref{lem:optbyt} and \cref{prop:optbyc} we have left with an optimization of the right hand side of \eqref{eq:pAz} written as $\langle c, G^{-1}c\rangle$, where it is enough to use only $ne+1$ distinct complex $z_i$ numbers. 
    When we consider only the minimal $ne+1$ number of $z_i$ factors, then $Z$ becomes the Vandermonde matrix, which is invertible.
     Then \eqref{eq:optcGc} can be rewritten as 
    \begin{equation}
    \begin{split}
        \langle c,G^{-1}c\rangle&= \langle e_1,{Z^*}^{-1}G^{-1}Z^{-1}e_1\rangle=\sum_{i=1}^{ne+1} \norm{A_1(z_i)\otimes\dots\otimes A_k(z_i)}^{2n} \prod_{l\neq i}\frac{\lvert z_l\rvert^2}{\lvert z_i-z_l\rvert^2}\\ &=
        \sum_{i=1}^{ne+1} \prod_{l\neq i} \sqrt[e]{\norm{A_1(z_i)\otimes\dots\otimes A_k(z_i)}^{2}}\frac{\lvert z_l\rvert^2}{\lvert z_i-z_l\rvert^2}.
    \end{split}
    \end{equation}
\end{proof}

\begin{corollary}\label{cor:finalexpcountable}
The immediate consequence of \cref{cor:finalcountable} is that the optimal error exponent (\cref{def:errorexp}) of the protocol described in \cref{sec:DegenLOCC}, can be written in a form of an optimization problem over a countable subset of $\complexes$ 
    \begin{equation}
    r_\textnormal{opt}=\limsup_{n\to\infty}\frac{1}{n}\log\inf_{\{z_i\}_i\subseteq \complexes}\sum_{i=1}^{ne+1} \prod_{l\neq i} \sqrt[e]{\norm{A_1(z_i)\otimes\dots\otimes A_k(z_i)}^{2}}\frac{\lvert z_l\rvert^2}{\lvert z_i-z_l\rvert^2}.
\end{equation}
\end{corollary}

Next we show that the countable subset $\{z_i\}_i$ can be chosen from an increasing series of compact sets $K_n\subseteq\complexes$. 

\begin{proposition}\label{prop:countableKtoC}
Assume that for $z\to 0$ we have $\norm{A_1(z)\otimes\dots\otimes A_k(z)}\to\infty$ and that there exist $C, d, b >0$ such that $\norm{A_1(z)\otimes\dots\otimes A_k(z)}\ge C\lvert z\rvert^e$ for any $\norm{z}\ge d$ and  $\norm{A_1(z)\otimes\dots\otimes A_k(z)}\ge b$ everywhere. Then
\begin{equation}
\begin{split}
&\limsup_{n\to\infty}\frac{1}{n}\log
    \inf_{\{z_i\}_i\subseteq \complexes}\sum_{i=1}^{ne+1} \prod_{l\neq i} \sqrt[e]{\norm{A_1(z_i)\otimes\dots\otimes A_k(z_i)}}\frac{\lvert z_l\rvert^2}{\lvert z_i-z_l\rvert^2}=\\
    \inf_{K\subseteq\complexes}&\limsup_{n\to\infty}\frac{1}{n}\log\inf_{\{z_i\}_i\subseteq K}\sum_{i=1}^{ne+1} \prod_{l\neq i} \sqrt[e]{\norm{A_1(z_i)\otimes\dots\otimes A_k(z_i)}}\frac{\lvert z_l\rvert^2}{\lvert z_i-z_l\rvert^2},
\end{split}
\end{equation}
where the first infimum is taken over the compact sets $K$, then the next is taken over every $n$-element subsets of $K$. 
\end{proposition}
\begin{proof}
In this proof we use the abbreviation $A(z)\coloneqq \norm{A_1(z)\otimes\dots\otimes A_k(z)}$.
One can see that the left hand side is trivially less than or equal to the right hand side. For the converse let $0\le m\le n$ and $z_0,\dots,z_{ne}\in\complexes$ ordered by magnitude, i.e., $\lvert z_0\rvert\le\lvert z_1\rvert\le\dots\le\lvert z_{ne}\rvert$. Let
\begin{equation}
M=\max_{i\in[0,ne]}\Bigg[\log A(z_i)+e\frac{1}{ne}\sum_{\substack{l=0 \\ l\neq i}}^{ne}\log\frac{\lvert z_l\rvert}{\lvert z_i-z_l\rvert}\Bigg].
\end{equation}

We consider three lower bounds.
First, we substitute $i=0$ instead the maximization:
\begin{equation}\label{eq:smallestzbound}
\begin{split}
M
 & \ge \log A(z_0)+e\frac{1}{ne}\sum_{l=1}^{ne}\log\frac{\lvert z_l\rvert}{\lvert z_0-z_l\rvert}  \\
 & \ge \log A(z_0)-e,
\end{split}
\end{equation}
using that
\begin{equation}
\frac{\lvert z_l\rvert}{\lvert z_0-z_l\rvert}\ge\frac{\lvert z_l\rvert}{\lvert z_0\rvert+\lvert z_l\rvert}\ge\frac{1}{2}.
\end{equation}
Since $\lim_{z\to 0}A(z)=\infty$, \eqref{eq:smallestzbound} implies that the maximum is unbounded unless $z_0$ is bounded away from $0$.

Next, we substitute $i=me$:
\begin{equation}
\begin{split}
M
 & \ge \log A(z_{me})+e\frac{1}{ne}\sum_{l=0}^{me-1}\log\frac{\lvert z_l\rvert}{\lvert z_{me}-z_l\rvert}+e\frac{1}{ne}\sum_{l=me+1}^{ne}\log\frac{\lvert z_l\rvert}{\lvert z_{me}-z_l\rvert}  \\
 & \ge \log A(z_{me})+e\frac{me}{ne}\log\frac{\lvert z_0\rvert}{\lvert z_{me}\rvert+\lvert z_0\rvert}-e\frac{ne-me}{ne}  \\
 & \ge \log \frac{A(z_{me})}{\lvert z_{me}\rvert^{\frac{m}{n}e}}+e\frac{m}{n}\log\lvert z_0\rvert-e.
\end{split}
\end{equation}
Since $A(z)\ge C\lvert z\rvert^e$ for large $\lvert z\rvert$ and $\lvert z_0\rvert$ is bounded from below, if we set $m=\lfloor(1-\epsilon)n\rfloor$ then $\lvert z_{me}\rvert$ must be bounded. In other words, $\{\nu_v\}_{n\in\naturals}$ is a tight family of measures.

Finally, we maximize over $i\in[0,me]$:
\begin{equation}
\begin{split}
M
 & \ge \max_{i\in[0,me]}\Bigg[\log A(z_i)+e\frac{1}{ne}\sum_{\substack{l=0 \\ l\neq i}}^{ne}\log\frac{\lvert z_l\rvert}{\lvert z_i-z_l\rvert}\Bigg]  \\
 & \ge \max_{i\in[0,me]}\Bigg[\log A(z_i)+e\frac{1}{ne}\sum_{\substack{l=0 \\ l\neq i}}^{me}\log\frac{\lvert z_l\rvert}{\lvert z_i-z_l\rvert}+e\frac{1}{ne}\sum_{l=me+1}^{ne}\log\frac{\lvert z_l\rvert}{\lvert z_i-z_l\rvert}\Bigg]  \\
 & \ge \max_{i\in[0,me]}\Bigg[\left(1-\frac{m}{n}\right)\log A(z_i)+\frac{m}{n}\Bigg(\log A(z_i)+e\frac{1}{me}\sum_{\substack{l=0 \\ l\neq i}}^{me}\log\frac{\lvert z_l\rvert}{\lvert z_i-z_l\rvert}\Bigg)-\left(1-\frac{m}{n}\right)e\Bigg]  \\
 & \ge \frac{m}{n}\max_{i\in[0,me]}\Bigg[\log A(z_i)+e\frac{1}{me}\sum_{\substack{l=0 \\ l\neq i}}^{me}\log\frac{\lvert z_l\rvert}{\lvert z_i-z_l\rvert}\Bigg]-\left(1-\frac{m}{n}\right)e+\left(1-\frac{m}{n}\right)\log \min_{z\in\complexes}A(z)
\end{split}
\end{equation}
Now set $m=\lfloor(1-\epsilon)n\rfloor$ and then $n\to\infty$. By the previous estimates, there are constants $0<C_1<C_2$ such that $C_1\le\lvert z_0\rvert\le\lvert z_1\rvert\le\dots\le\lvert z_{me}\rvert\le C_2$. The expression in the brackets takes the form of $M$, but this time the maximum is taken by respecting these confinements. Let us denote this expression by $M_c$, and write
\begin{equation}
M - M_c \ge \frac{1}{(1-\epsilon)}\left( \epsilon e+\epsilon \log \min_{z\in\complexes}A(z)\right).
\end{equation}
As $\epsilon\to 0$ this gives $0$, implying that $M \ge M_c$.

\end{proof}

\begin{corollary}
    Note that if the Laurent polynomial $A_1(z)\otimes\dots\otimes A_l(z)$ contains both positive and negative powers of $z$ and it is nowhere zero then the conditions of \cref{prop:countableKtoC} are satisfied for $\norm{A_1(z)\otimes\dots\otimes A_l(z)}$. By \cref{rem:degenInterestingcases} these are the only interesting cases to consider.
    Then by the previous proposition the optimal error exponent in $\cref{cor:finalexpcountable}$ can be written as 
\begin{equation}\label{eq:roptdiscretecompact}
    r_\textnormal{opt}=\inf_K\limsup_{n\to\infty}\frac{1}{n}\log\inf_{\{z_i\}_i\subseteq K}\sum_{i=1}^{ne+1} \prod_{l\neq i} \sqrt[e]{\norm{A_1(z_i)\otimes\dots\otimes A_l(z_i)}^{2}}\frac{\lvert z_l\rvert^2}{\lvert z_i-z_l\rvert^2},
\end{equation}
where the first infimum is taken over the compact sets $K$ then the next is taken over the $n$-element subsets of $K$. 
\end{corollary}

\section{Limit of optimal probabilities}\label{sec:limitofoptprob}

In the following we show that the optimization in \eqref{eq:roptdiscretecompact} can be written as an optimization of an integral over a probability measure. Although this problem resembles the topic of potential theory, and we aim for a similar equality as the one between the transfinite diameter and the logarithmic potential we introduce in \cref{subsec:transfiniteNpot}, the same approaches and facts are not directly applicable here.

Let $K\subseteq O\subseteq\complexes$ be such that $K$ is compact and $O$ is open, and let $w_1,w_2:O\to\positivereals$ be continuous functions.
We define
\begin{equation}\label{eq:deltaKw1w2def}
\delta_{K,w_1,w_2}=\limsup_{n\to\infty}\frac{1}{n}\log\max_{z_0,\dots,z_n\in K}\min_{0\le i\le n}\prod_{\substack{l=0 \\ l\neq i}}^n\lvert z_i-z_l\rvert w_1(z_i)w_2(z_l).
\end{equation}
For $\epsilon\ge 0$ let
\begin{equation}
K_\epsilon=\setbuild{z\in\complexes}{\distance(z,K)\le\epsilon},
\end{equation}
where $\distance(z,K)=\min_{w\in K}\lvert z-w\rvert$. Note that $K_\epsilon\subseteq O$ for all sufficiently small $\epsilon$.
Our goal is to prove that
\begin{equation}\label{eq:supintbetweendelta}
\delta_{K,w_1,w_2}\le\sup_{\sigma\in\distributions(K)}\inf_{z\in\support(\sigma)}\int_K\log\lvert z-t\rvert w_1(z)w_2(t)\ed\sigma(t)\le\delta_{K_\epsilon,w_1,w_2}
\end{equation}
and for all $\epsilon>0$ such that $K_\epsilon\subseteq O$ (with the logarithm extended as $\log0=-\infty$).

\subsection{Lower bound}

To prove the second inequality in \eqref{eq:supintbetweendelta}, we need to show that for all $\sigma\in\distributions(K)$ and small $\epsilon>0$ the inequality
\begin{equation}\label{eq:deltalowerbound}
\delta_{K_\epsilon,w_1,w_2}\ge\inf_{z\in\support(\sigma)}\int_K\log\lvert z-t\rvert w_1(z)w_2(t)\ed\sigma(t)
\end{equation}
holds. In the first step, we will replace $\sigma$ with a coarse-grained measure $\sigma_a$ with slightly larger support, then we approximate it with a sequence of normalized counting measures, as illustrated in \cref{fig:discretization}.

We note that the inequality is vacuous if $\sigma(\{\zeta\})>0$ for some $\zeta\in K$, since then for $z=\zeta$ the integrand is $-\infty$ on a set of positive measure (namely at $t=\zeta$). Therefore in the following we assume $\sigma(\{\zeta\})=0$ for all $\zeta\in K$.

\subsubsection{Coarse graining}

For $a>0$ we define the following measurable partition $\mathcal{A}_a=\setbuild{A_{a,x,y}}{x,y\in\integers}$ of $\complexes$, where
\begin{equation}
A_{a,x,y}=\setbuild{z\in\complexes}{(x-\textstyle\frac{1}{2})a\le\Re z<(x+\textstyle\frac{1}{2})a,(y-\textstyle\frac{1}{2})a\le\Im z<(y+\textstyle\frac{1}{2})a}.
\end{equation}
We define the coarse-grained measure $\sigma_a$ as
\begin{equation}
\sigma_a(E)=\sum_{x,y\in\integers}\sigma(A_{a,x,y})\frac{1}{a^2}\lambda(E\cap A_{a,x,y})
\end{equation}
for all Borel sets $E$, where $\lambda$ denotes the (two-dimensional) Lebesgue measure on $\complexes$. Note that $\support\sigma\subseteq\support\sigma_a\subseteq K_\epsilon$, and for all $0<a\le\epsilon/\sqrt{2}$.

\begin{figure}
\colorlet{regioncolor}{black!65}
\begin{tikzpicture}
\fill[regioncolor] (5:2.8) arc[start angle=5, end angle=75, radius=2.8] -- (75:1.5) arc[start angle=75, end angle=5, radius=1.5] -- cycle;
\draw[help lines] (-0.5,-0.5) grid (3.5,3.5);
\draw[<->] (0,-0.3)--(1,-0.3) node[below,midway] {$a$};
\end{tikzpicture}
\hspace{\stretch{1}}
\begin{tikzpicture}
\foreach \pos / \percentage in {(1,0)/54, (2,0)/57, (0,1)/38, (1,1)/99, (2,1)/34, (0,2)/25, (1,2)/34}
\fill[regioncolor!\percentage] \pos rectangle +(1,1);

\draw[help lines] (-0.5,-0.5) grid (3.5,3.5);
\draw[<->] (0,-0.3)--(1,-0.3) node[below,midway] {$a$};
\end{tikzpicture}
\hspace{\stretch{1}}
\begin{tikzpicture}
\foreach \pos / \t in {(1,0)/4, (2,0)/5, (0,1)/4, (1,1)/6, (2,1)/4, (0,2)/3, (1,2)/4} 
{
  \draw[very thin,black!25] \pos grid[step=1/\t] +(1,1);
\foreach \i in {1,...,\t}
\foreach \j in {1,...,\t}
  \fill \pos +({(\i-0.5)/\t},{(\j-0.5)/\t}) circle (0.2mm);
}

\draw[help lines] (-0.5,-0.5) grid (3.5,3.5);
\draw[<->] (0,-0.3)--(1,-0.3) node[below,midway] {$a$};
\end{tikzpicture}
\caption{The first image shows the original probability distribution and the chosen lattice structure. In the middle we see the coarse grained distribution. The last image shows the discretization of the coarse grained distribution.}\label{fig:discretization}
\end{figure}
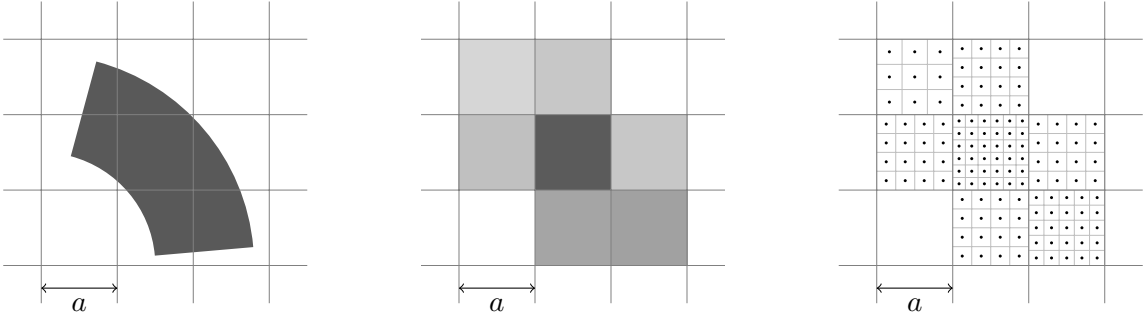

\begin{lemma}\label{lem:logintegralbound}
The function $f:\reals^2\to\reals$,
\begin{equation}
f(u,v)=\frac{1}{2}\int_{-\frac{1}{2}}^{\frac{1}{2}}\int_{-\frac{1}{2}}^{\frac{1}{2}}\log\frac{\left(\lvert u\rvert+\frac{1}{2}\right)^2+\left(\lvert v\rvert+\frac{1}{2}\right)^2}{\left(u-\xi\right)^2+\left(v-\eta\right)^2}\ed\xi\ed\eta
\end{equation}
is bounded and satisfies $f(u,v)\to 0$ as $\sqrt{u^2+v^2}\to\infty$.
\end{lemma}
\begin{proof}
Note that the integrand is non-negative for any $\xi,\eta\in [-\frac{1}{2},\frac{1}{2}]$. Then for the boundedness it is enough to show that it is bounded from above. Also it is enough to consider the denominator in the argument of the logarithm, because its numerator is constant as a function of $\eta$ and $\xi$.
First we use the parametrization $\xi'\coloneqq u-\xi$, $\eta'\coloneqq v-\eta$ and $r^2\coloneqq \xi^{'2}+\eta^{'2}$.
The integral can be split into a part with $r> \frac{1}{2}$ and a part with $r\le\frac{1}{2}$. In the first case the integrand is bounded by $-\log{\frac{1}{4}}$. In the second case we use polar coordinates to evaluate the integral
\begin{equation}
\int_{0}^{2\pi}\int_{0}^{R=\frac{1}{2}}-r\log\left(r^2\right)\ed r\ed\varphi = 
\pi R^2 (1-2\log R)=\frac{\pi}{4} (1-2\log \frac{1}{2})
\end{equation}
To show the limit we rewrite the argument of the logarithm by using the notation $r^2\coloneqq \lvert u\rvert^2+\lvert v\rvert^2$,
\begin{equation}
\begin{split}
    \frac{\lvert u\rvert^2+\lvert u\rvert+\frac{1}{4}+\lvert v\rvert^2+\lvert v\rvert+\frac{1}{4}}{
    \lvert u\rvert^2-2u\lvert\xi\rvert +\lvert\xi\rvert^2 +\lvert v\rvert^2-2v\lvert\eta\rvert +\lvert\eta\rvert^2}&=
    \frac{1+\frac{\lvert u\rvert+\lvert v\rvert+\frac{1}{2}}{r^2}}{
    1-\frac{2u\lvert\xi\rvert +\lvert\xi\rvert^2+2v\lvert\eta\rvert +\lvert\eta\rvert^2}{r^2}}\\
    &\le\frac{1+\frac{\lvert u\rvert+\lvert v\rvert+\frac{1}{2}}{r^2}}{
    1-\frac{\lvert u\rvert \lvert v\rvert }{r^2}}\\
    &\le\frac{1+\frac{\sqrt{2}r+\frac{1}{2}}{r^2}}{
    1-\frac{\sqrt{2}r }{r^2}}.
    \end{split}
\end{equation}
The logarithm of the right hand side vanishes as $r\to \infty$. By this and the non-negativity of the integrand $f(u,v)\to 0$.
\end{proof}

\begin{lemma}\label{lem:smallballmeasure}
Let $g(r)=\sup\setbuild{\sigma(\ball{r}{\zeta})}{\zeta\in K}$. Then $\lim_{r\to 0}g(r)=0$.
\end{lemma}
\begin{proof}
The function $g$ is monotone increasing, therefore the limit exists.

For $n\in\positiveintegers$ let $\zeta_n\in K$ such that $g(\frac{1}{n})\le\sigma(\ball{r}{\zeta_n})+\frac{1}{n}$. Let $(\zeta_{n_i})_{i\in\naturals}$ be a convergent subsequence and let $\zeta\in K$ be its limit. We may further assume that $\lvert\zeta_{n_i}-\zeta\rvert$ is monotone decreasing. Then $\ball{\frac{1}{n_i}}{\zeta_{n_i}}\subseteq\ball{\frac{1}{n}+\lvert\zeta_{n_i}-\zeta\rvert}{\zeta}$ and the right hand side is a decreasing sequence, therefore
\begin{equation}
\begin{split}
0
 & = \sigma(\{\zeta\})  \\
 & = \sigma\left(\bigcap_{i\in\naturals}\ball{\frac{1}{n}+\lvert\zeta_{n_i}-\zeta\rvert}{\zeta}\right)  \\
 & = \lim_{i\to\infty}\sigma\left(\ball{\frac{1}{n}+\lvert\zeta_{n_i}-\zeta\rvert}{\zeta}\right)  \\
 & \ge \limsup_{i\to\infty}\sigma(\ball{\frac{1}{n_i}}{\zeta_{n_i}})  \\
 & \ge \lim_{i\to\infty}\left(g\left(\frac{1}{n_i}\right)-\frac{1}{n_i}\right)  \\
 & = \lim_{i\to\infty}g\left(\frac{1}{n_i}\right)  \\
 & = \lim_{r\to 0}g(r).
\end{split}
\end{equation}
\end{proof}

\begin{proposition}\label{prop:coarsegrainedlowerbound}
\begin{multline}
\liminf_{a\to 0}\inf_{z\in\support(\sigma_a)}\int_O\log\lvert z-t\rvert w_1(z)w_2(t)\ed\sigma_a(t)  \\  \ge\inf_{z\in\support(\sigma)}\int_K\log\lvert z-t\rvert w_1(z)w_2(t)\ed\sigma(t).
\end{multline}
\end{proposition}
\begin{proof}
Let $R>0$ be arbitrary and $\epsilon>0$ such that $K_\epsilon\subseteq O$ and $a\le\epsilon/\sqrt{2}$ so that $\support(\sigma_a)\subseteq K_\epsilon$. Since $K_\epsilon$ is compact, $\log w_1$ and $\log w_2$ are uniformly continuous on $K_\epsilon$.

For an arbitrary $z\in\support(\sigma_a)$ let $\tilde{z}$ be one of the closest points to $z$ in $\support(\sigma)$. Then
\begin{multline}\label{eq:coarsegrainedintegraldifference}
\int_K\log\lvert \tilde{z}-t\rvert w_1(\tilde{z})w_2(t)\ed\sigma(t)-\int_O\log\lvert z-t\rvert w_1(z)w_2(t)\ed\sigma_a(t)  \\
\begin{aligned}
 & =
\log w_1(\tilde{z})-\log w_1(z)  \\
 & + \int_K\log\lvert \tilde{z}-t\rvert\ed\sigma(t)-\int_K\log\lvert z-t\rvert\ed\sigma(t)  \\
 & + \int_K\log\lvert z-t\rvert\ed\sigma(t)-\int_O\log\lvert z-t\rvert\ed\sigma_a(t)  \\
 & +\int_K \log w_2(t)\ed\sigma(t)-\int_O \log w_2(t)\ed\sigma_a(t).
\end{aligned}
\end{multline}
Since $\lvert\tilde{z}-z\rvert<\sqrt{2}a$ and $\log w_1$ is uniformly continuous, the first difference vanishes uniformly in $z$ as $a\to 0$.

The second difference can be bounded as
\begin{multline}
\int_K\log\lvert \tilde{z}-t\rvert\ed\sigma(t)-\int_K\log\lvert z-t\rvert\ed\sigma(t)  \\
\begin{aligned}
 & = \int_K\log\frac{\lvert \tilde{z}-z+z-t\rvert}{\lvert z-t\rvert}\ed\sigma(t)  \\
 & \le \int_K\log\left(1+\frac{\lvert \tilde{z}-z\rvert}{\lvert z-t\rvert}\right)\ed\sigma(t)  \\
 & = \int_{K\cap\ball{R}{z}}\log\left(1+\frac{\lvert \tilde{z}-z\rvert}{\lvert z-t\rvert}\right)\ed\sigma(t)+\int_{K\setminus\ball{R}{z}}\log\left(1+\frac{\lvert \tilde{z}-z\rvert}{\lvert z-t\rvert}\right)\ed\sigma(t)  \\
 & \le \sigma(\ball{R}{z})+\frac{1}{\ln 2}\int_{K\setminus\ball{R}{z}}\frac{\lvert \tilde{z}-z\rvert}{\lvert z-t\rvert}\ed\sigma(t)  \\
 & \le g(R)+\frac{\sqrt{2}}{\ln 2}\frac{a}{R},
\end{aligned}
\end{multline}
in the second inequality using that $\lvert \tilde{z}-z\rvert\le\lvert z-t\rvert$ holds by the choice of $\tilde{z}$, and $\ln(1+x)\le x$.

For the third difference in \eqref{eq:coarsegrainedintegraldifference} we use
\begin{equation}
\begin{split}
& \int_K\log\lvert z-t\rvert\ed\sigma(t)-\int_O\log\lvert z-t\rvert\ed\sigma_a(t)  \\
 & = \sum_{x,y\in\integers}\left[\int_{A_{a,x,y}}\log\lvert z-t\rvert\ed\sigma(t)-\int_{A_{a,x,y}}\log\lvert z-t\rvert\ed\sigma_a(t)\right]  \\
 & = \sum_{x,y\in\integers}\sigma(A_{a,x,y})\left[\max_{t\in A_{a,x,y}}\log\lvert z-t\rvert-\int_{A_{a,x,y}}\log\lvert z-t\rvert\ed\sigma_a(t)\right]  \\
 & = \sum_{x,y\in\integers}\sigma(A_{a,x,y})\left[\max_{t\in A_{a,x,y}}\log\lvert z-t\rvert-\frac{1}{\sigma(A_{a,x,y})}\int_{A_{a,x,y}}\log\lvert z-t\rvert\ed\sigma_a(t)\right]  \\
\end{split}
\end{equation}
In each term we make the substitution $t=(x+\xi)a+(y+\eta)ai$ so that $A_{a,x,y}$ is parametrized by $\xi,\eta\in[-\frac{1}{2},\frac{1}{2}]$. Then
\begin{equation}
\begin{split}
& \max_{t\in A_{a,x,y}}\log\lvert z-t\rvert-\frac{1}{\sigma(A_{a,x,y})}\int_{A_{a,x,y}}\log\lvert z-t\rvert\ed\sigma_a(t)  \\
 & = \max_{\xi,\eta\in[-\frac{1}{2},\frac{1}{2}]}\log\lvert z-((x+\xi)a+(y+\eta)ai)\rvert  \\ &\quad-\int_{-\frac{1}{2}}^{\frac{1}{2}}\int_{-\frac{1}{2}}^{\frac{1}{2}}\log\lvert z-((x+\xi)a+(y+\eta)ai)\rvert\ed\xi\ed\eta  \\
 & = \log\sqrt{\left(\left\lvert\frac{\Re z}{a}-x\right\rvert+\frac{1}{2}\right)^2+\left(\left\lvert\frac{\Im z}{a}-y\right\rvert+\frac{1}{2}\right)^2}  \\ &\quad-\int_{-\frac{1}{2}}^{\frac{1}{2}}\int_{-\frac{1}{2}}^{\frac{1}{2}}\log\sqrt{\left(\frac{\Re z}{a}-x-\xi\right)^2+\left(\frac{\Im z}{a}-y-\eta\right)^2}\ed\xi\ed\eta  \\
 & = f\left(\frac{\Re z}{a}-x,\frac{\Im z}{a}-y\right).
\end{split}
\end{equation}

We split the sum over $x$ and $y$ into terms with $\lvert z-(x+yi)a\rvert<R$ and $\lvert z-(x+yi)a\rvert\ge R$.
\begin{equation}
\begin{split}
\sum_{x,y\in\integers}\sigma(A_{a,x,y})f\left(\frac{\Re z}{a}-x,\frac{\Im z}{a}-y\right)
 & = \sum_{\substack{x,y\in\integers  \\  \lvert z-(x+yi)a\rvert<R}}\sigma(A_{a,x,y})f\left(\frac{\Re z}{a}-x,\frac{\Im z}{a}-y\right)  \\  &\quad+\sum_{\substack{x,y\in\integers  \\  \lvert z-(x+yi)a\rvert\ge R}}\sigma(A_{a,x,y})f\left(\frac{\Re z}{a}-x,\frac{\Im z}{a}-y\right)  \\
 & \le \sigma(\ball{R+a/\sqrt{2}}{z})\norm{f}+\sup_{\substack{u,v\in\reals  \\  u^2+v^2\ge \frac{R^2}{a^2}}}f(u,v)  \\
 & \le g(R+a/\sqrt{2})\norm{f}+\sup_{\substack{u,v\in\reals  \\  u^2+v^2\ge \frac{R^2}{a^2}}}f(u,v)
\end{split}
\end{equation}

The fourth difference in \eqref{eq:coarsegrainedintegraldifference} does not depend on $z$ and can be bounded as
\begin{equation}
\begin{split}
\int_K\log w_2(t)\ed\sigma(t)-\int_O\log w_2(t)\ed\sigma_a(t)
 & = \sum_{x,y\in\integers}\left[\int_{A_{a,x,y}} \log w_2(t)\ed\sigma(t)-\int_{A_{a,x,y}} \log w_2(t)\ed\sigma_a(t)\right]  \\
 &\le \sum_{x,y\in\integers}\sigma(A_{a,x,y})\left[\max_{t\in A_{a,x,y}}\log w_2(t)-\min_{t\in A_{a,x,y}}\log w_2(t)\right]  \\
 &\le \max_{x,y\in\integers}\left[\max_{t\in A_{a,x,y}}\log w_2(t)-\min_{t\in A_{a,x,y}}\log w_2(t)\right].
\end{split}
\end{equation}
Since the diameter of $A_{a,x,y}$ is $\sqrt{2}a$ and $\log w_2$ is uniformly continuous on $K_\epsilon$, the upper bound vanishes as $a\to 0$.

Combining the bounds, we find in the limit that
\begin{multline}
\liminf_{a\to 0}\inf_{z\in\support(\sigma_a)}\int_O\log\lvert z-t\rvert w_1(z)w_2(t)\ed\sigma_a(t)  \\
\ge\inf_{z\in\support(\sigma_a)}\int_K\log\lvert \tilde{z}(z)-t\rvert w_1(\tilde{z}(z))w_2(t)\ed\sigma(t)  \\ -\left[g(R)+\lim_{\rho\to R+}g(\rho)\norm{f}+\liminf_{a\to 0}\sup_{\substack{u,v\in\reals  \\  u^2+v^2\ge \frac{R^2}{a^2}}}f(u,v)\right]\\
\ge\inf_{\tilde{z}\in\support(\sigma)}\int_K\log\lvert \tilde{z}-t\rvert w_1(z)w_2(t)\ed\sigma(t)  -\left[g(R)+\lim_{\rho\to R+}g(\rho)\norm{f}\right].
\end{multline}
In the second inequality we used that for the map $\tilde{z}: \support(\sigma_a)\to\support(\sigma)$ we have $\Im(\tilde{z}) \subseteq \support(\sigma)$ and also that by \cref{lem:logintegralbound} the last term in the square brackets vansihes as $a\to 0$.
Finally, in the last line the term in square brackets vanishes as $R\to 0$ by \cref{lem:smallballmeasure}. 
\end{proof}

\subsubsection{Discrete approximation}

In this section we show that the integral with respect to the coarse-grained measure $\sigma_a$ is a lower bound on $\delta_{K_\epsilon,w_1,w_2}$ when $a$ is sufficiently small (so that $\support\sigma_a\subseteq K_\epsilon$). To this end, we approximate $\sigma_a$ with the normalized counting measure of a set of points distributed on square lattices within each subset $A_{a,x,y}$, as illustrated in \cref{fig:discretization}. These points will be used to bound the maximum over $z_0,\dots,z_n$ in \eqref{eq:deltaKw1w2def} from below.

With $a>0$ fixed and $N\in\naturals$, let $t_{x,y}=\lceil\sqrt{\sigma(A_{a,x,y})N}\rceil$ for each $x,y\in\integers$, and consider the following $N\le\sum_{x,y\in\integers}t_{x,y}^2$ complex numbers:
\begin{equation}
z^{(N)}_{x,y,i,j}=\left(x-\frac{1}{2}+\frac{i-\frac{1}{2}}{t_{x,y}}\right)a+\left(y-\frac{1}{2}+\frac{j-\frac{1}{2}}{t_{x,y}}\right)ai,
\end{equation}
where $x,y\in\integers$, $i,j\in[t_{x,y}]$. Note that $z^{(N)}_{x,y,i,j}\in\support(\sigma_a)$.

The next lemma relates the value of the integral over a lattice square around $z^{(N)}_{x,y,i,j}$ to the corresponding term in \eqref{eq:deltaKw1w2def}.
\begin{lemma}\label{lem:logintegralonsmallsquare}
Let $w\in\complexes$.
\begin{enumerate}
\item If $w\neq 0$, then
\begin{equation}
\int_{-\frac{1}{2}}^{\frac{1}{2}}\int_{-\frac{1}{2}}^{\frac{1}{2}}\log\left\lvert w-\frac{a}{t_{x,y}}(\xi+\eta i)\right\rvert \ed\xi\ed\eta\le\log\lvert w\rvert+\frac{1}{12\ln 2}\frac{a^2}{t_{x,y}^2}\frac{1}{\lvert w\rvert^2}.
\end{equation}
\item\label{it:logintegralonsmallsquare2} If $w=0$ and $\frac{a}{t_{x,y}}\le\sqrt{2}$, then
\begin{equation}
\int_{-\frac{1}{2}}^{\frac{1}{2}}\int_{-\frac{1}{2}}^{\frac{1}{2}}\log\left\lvert w-\frac{a}{t_{x,y}}(\xi+\eta i)\right\rvert \ed\xi\ed\eta\le 0.
\end{equation}
\end{enumerate}
\end{lemma}
\begin{proof}
In the first part, we use that $\log(1+u)\le\frac{1}{\ln 2}u$ to bound the integrand as follows:
\begin{multline}
\int_{-\frac{1}{2}}^{\frac{1}{2}}\int_{-\frac{1}{2}}^{\frac{1}{2}}\log\left\lvert w-\frac{a}{t_{x,y}}(\xi+\eta i)\right\rvert \ed\xi\ed\eta  \\
\begin{aligned}
 & = \log\lvert w\rvert+\frac{1}{2}\int_{-\frac{1}{2}}^{\frac{1}{2}}\int_{-\frac{1}{2}}^{\frac{1}{2}}\log\frac{(\Re w-\frac{a}{t_{x,y}}\xi)^2+(\Im w-\frac{a}{t_{x,y}}\eta)^2}{(\Re w)^2+(\Im w)^2} \ed\xi\ed\eta  \\
 & = \log\lvert w\rvert+\frac{1}{2}\int_{-\frac{1}{2}}^{\frac{1}{2}}\int_{-\frac{1}{2}}^{\frac{1}{2}}\log\left(1+\frac{-2\frac{a}{t_{x,y}}\xi\Re w-2\frac{a}{t_{x,y}}\eta\Im w+\frac{a^2}{t_{x,y}^2}(\xi^2+\eta^2)}{(\Re w)^2+(\Im w)^2}\right) \ed\xi\ed\eta  \\
 & \le \log\lvert w\rvert+\frac{1}{2\ln 2}\int_{-\frac{1}{2}}^{\frac{1}{2}}\int_{-\frac{1}{2}}^{\frac{1}{2}}\frac{-2\frac{a}{t_{x,y}}\xi\Re w-2\frac{a}{t_{x,y}}\eta\Im w+\frac{a^2}{t_{x,y}^2}(\xi^2+\eta^2)}{\lvert w\rvert^2} \ed\xi\ed\eta  \\
 & = \log\lvert w\rvert+\frac{1}{12\ln 2}\frac{a^2}{t_{x,y}^2}\frac{1}{\lvert w\rvert^2}.
\end{aligned}
\end{multline}

In \ref{it:logintegralonsmallsquare2}, the argument of the logarithm can be bounded as
\begin{equation}
\left\lvert\frac{a}{t_{x,y}}(\xi+\eta i)\right\rvert\le\frac{a}{t_{x,y}}\sqrt{\left(\frac{1}{2}\right)^2+\left(\frac{1}{2}\right)^2}\le 1,
\end{equation}
therefore the integrand is negative.
\end{proof}

\begin{proposition}\label{prop:deltalowerbound}
For all $a,\epsilon>0$ satisfying $\support(\sigma_a)\subseteq K_\epsilon$, the inequality
\begin{equation}
\delta_{K_\epsilon,w_1,w_2}\ge\inf_{z\in\support(\sigma_a)}\int_O\log\lvert z-t\rvert w_1(z)w_2(t)\ed\sigma_a(t).
\end{equation}
holds.
\end{proposition}
\begin{proof}
Let $N$ be so large that $\frac{a}{t_{x,y}}\le\sqrt{2}$ for all $x,y\in\integers$ satisfying $\sigma(A_{a,x,y})\neq 0$. Let $x_0,y_0\in\integers$ and $i_0,j_0\in[t_{x,y}]$ such that
\begin{equation}
\prod_{\substack{x,y,i,j  \\  (x,y,i,j)\neq(x_0,y_0,i_0,j_0)}}\lvert z^{(N)}_{x_0,y_0,i_0,j_0}-z^{(N)}_{x,y,i,j}\rvert w_1(z^{(N)}_{x_0,y_0,i_0,j_0})w_2(z^{(N)}_{x,y,i,j})
\end{equation}
is minimal. Then
\begin{multline}\label{eq:integralchunks}
\inf_{z\in\support(\sigma_a)}\int_O\log\lvert z-t\rvert w_1(z)w_2(t)\ed\sigma_a(t)  \\
\begin{aligned}
 & \le \int_O\log\lvert z^{(N)}_{x_0,y_0,i_0,j_0}-t\rvert w_1(z^{(N)}_{x_0,y_0,i_0,j_0})w_2(t)\ed\sigma_a(t)  \\
 & = \log w_1(z^{(N)}_{x_0,y_0,i_0,j_0})  \\ &\quad +\sum_{x,y\in\integers}\sigma(A_{a,x,y})\sum_{i,j=1}^{t_{x,y}}\frac{1}{t_{x,y}^2}\int_{-\frac{1}{2}}^{\frac{1}{2}}\int_{-\frac{1}{2}}^{\frac{1}{2}}\log\left\lvert z^{(N)}_{x_0,y_0,i_0,j_0}-\left(z^{(N)}_{x,y,i,j}+\frac{a}{t_{x,y}}(\xi+\eta i)\right)\right\rvert \ed\xi\ed\eta  \\ &\quad +\sum_{x,y\in\integers}\sigma(A_{a,x,y})\sum_{i,j=1}^{t_{x,y}}\frac{1}{t_{x,y}^2}\int_{-\frac{1}{2}}^{\frac{1}{2}}\int_{-\frac{1}{2}}^{\frac{1}{2}}\log w_2\left(z^{(N)}_{x,y,i,j}+\frac{a}{t_{x,y}}(\xi+\eta i)\right)\ed\xi\ed\eta
\end{aligned}
\end{multline}

The second term in \eqref{eq:integralchunks} can be bounded as
\begin{multline}
\sum_{x,y\in\integers}\sigma(A_{a,x,y})\sum_{i,j=1}^{t_{x,y}}\frac{1}{t_{x,y}^2}\int_{-\frac{1}{2}}^{\frac{1}{2}}\int_{-\frac{1}{2}}^{\frac{1}{2}}\log\left\lvert z^{(N)}_{x_0,y_0,i_0,j_0}-\left(z^{(N)}_{x,y,i,j}+\frac{a}{t_{x,y}}(\xi+\eta i)\right)\right\rvert \ed\xi\ed\eta  \\
  \le \sum_{\substack{x,y,i,j  \\  (x,y,i,j)\neq(x_0,y_0,i_0,j_0)}}\left[\sigma(A_{a,x,y})\frac{1}{t_{x,y}^2}\log\lvert z^{(N)}_{x_0,y_0,i_0,j_0}-z^{(N)}_{x,y,i,j}\rvert+\frac{\sigma(A_{a,x,y})}{12\ln 2}\frac{a^2}{t_{x,y}^4}\frac{1}{\lvert z^{(N)}_{x_0,y_0,i_0,j_0}-z^{(N)}_{x,y,i,j}\rvert^2}\right]  \\
  \le \sum_{\substack{x,y,i,j  \\  (x,y,i,j)\neq(x_0,y_0,i_0,j_0)}}\left[\frac{1}{N}\log\lvert z^{(N)}_{x_0,y_0,i_0,j_0}-z^{(N)}_{x,y,i,j}\rvert+\frac{\sigma(A_{a,x,y})}{12\ln 2}\frac{a^2}{t_{x,y}^4}\frac{1}{\lvert z^{(N)}_{x_0,y_0,i_0,j_0}-z^{(N)}_{x,y,i,j}\rvert^2}\right]
\end{multline}
With $R=n^{-1/4}$ we split the sum of the inverse quadratic terms as follows, using that the smallest distance of any pair of points is $\min_{x,y}\frac{a}{t_{x,y}}$:
\begin{multline}
\sum_{\substack{x,y,i,j  \\  (x,y,i,j)\neq(x_0,y_0,i_0,j_0)}}\frac{\sigma(A_{a,x,y})}{t_{x,y}^4}\frac{1}{\lvert z^{(N)}_{x_0,y_0,i_0,j_0}-z^{(N)}_{x,y,i,j}\rvert^2}  \\
\begin{aligned}
 & = \sum_{\substack{x,y,i,j  \\  0<\lvert z^{(N)}_{x_0,y_0,i_0,j_0}-z^{(N)}_{x,y,i,j}\rvert<R}}\frac{\sigma(A_{a,x,y})}{t_{x,y}^4}\frac{1}{\lvert z^{(N)}_{x_0,y_0,i_0,j_0}-z^{(N)}_{x,y,i,j}\rvert^2}  \\  &\quad+\sum_{\substack{x,y,i,j  \\  R\le\lvert z^{(N)}_{x_0,y_0,i_0,j_0}-z^{(N)}_{x,y,i,j}\rvert}}\frac{\sigma(A_{a,x,y})}{t_{x,y}^4}\frac{1}{\lvert z^{(N)}_{x_0,y_0,i_0,j_0}-z^{(N)}_{x,y,i,j}\rvert^2}  \\
 & \le \frac{\max_{x,y}t_{x,y}^2}{a^2}\sum_{\substack{x,y,i,j  \\  0<\lvert z^{(N)}_{x_0,y_0,i_0,j_0}-z^{(N)}_{x,y,i,j}\rvert<R}}\frac{\sigma(A_{a,x,y})}{t_{x,y}^4}+\frac{1}{R^2}\sum_{\substack{x,y,i,j  \\  R\le\lvert z^{(N)}_{x_0,y_0,i_0,j_0}-z^{(N)}_{x,y,i,j}\rvert}}\frac{\sigma(A_{a,x,y})}{t_{x,y}^4}  \\
 & \le \frac{\max_{x,y}t_{x,y}^2}{a^2}\sum_{\substack{x,y,i,j  \\  0<\lvert z^{(N)}_{x_0,y_0,i_0,j_0}-z^{(N)}_{x,y,i,j}\rvert<R}}\frac{\sigma(A_{a,x,y})}{t_{x,y}^4}+\frac{1}{R^2}\sum_{x,y}\frac{\sigma(A_{a,x,y})}{t_{x,y}^2}  \\
\end{aligned}
\end{multline}
Around each point with a given $x,y$, we can place a square of side length $\frac{a}{t_{x,y}}$, and these squares are disjoint. The total area of these squares with centers within distance $R$ from a given point is at most $\left(R+\frac{1}{\sqrt{2}}\frac{a}{t_{x,y}}\right)^2\pi$, therefore the number of such points is at most
\begin{equation}
\frac{t_{x,y}^2}{a^2}\left(R+\frac{1}{\sqrt{2}}\frac{a}{t_{x,y}}\right)^2\pi.
\end{equation}
Using this,
\begin{multline}
\frac{\max_{x,y}t_{x,y}^2}{a^2}\sum_{\substack{x,y,i,j  \\  0<\lvert z^{(N)}_{x_0,y_0,i_0,j_0}-z^{(N)}_{x,y,i,j}\rvert<R}}\frac{\sigma(A_{a,x,y})}{t_{x,y}^4}  \\
\begin{aligned}
 & \le\frac{\max_{x,y}t_{x,y}^2}{a^2}\sum_{x,y}\frac{\sigma(A_{a,x,y})}{t_{x,y}^4}\frac{t_{x,y}^2}{a^2}\left(R+\frac{1}{\sqrt{2}}\frac{a}{t_{x,y}}\right)^2\pi  \\
 & \le\frac{\max_{x,y}t_{x,y}^2}{a^4\min_{x,y}t_{x,y}^2}\left(R+\frac{1}{\sqrt{2}}\frac{a}{\min_{x,y}t_{x,y}}\right)^2\pi,
\end{aligned}
\end{multline}
which is $O(n^{-1/2})$.

The last term in \eqref{eq:integralchunks} satisfies
\begin{multline}
\sum_{x,y\in\integers}\sigma(A_{a,x,y})\sum_{i,j=1}^{t_{x,y}}\frac{1}{t_{x,y}^2}\int_{-\frac{1}{2}}^{\frac{1}{2}}\int_{-\frac{1}{2}}^{\frac{1}{2}}\log w_2\left(z^{(N)}_{x,y,i,j}+\frac{a}{t_{x,y}}(\xi+\eta i)\right)\ed\xi\ed\eta  \\
\le \frac{1}{N}\sum_{x,y\in\integers}\sum_{i,j=1}^{t_{x,y}}\log w_2(z^{(N)}_{x,y,i,j})+\max\setbuild{\lvert\log w_2(z)-\log w_2(z')\rvert}{z, z' \in \complexes: \lvert z-z'\rvert\le\frac{1}{\sqrt{2}}\frac{a}{\min_{x,y}t_{x,y}}},
\end{multline}
and the maximum on the right hand side vanishes as $N\to\infty$ because $\log w_2$ is uniformly continuous and $\min_{x,y}t_{x,y}\to\infty$.

Combining the estimates, we obtain
\begin{equation}
\begin{split}
\delta_{K_\epsilon,w_1,w_2}
 & \ge \limsup_{N\to\infty}\frac{1}{n-1}\log\prod_{\substack{x,y,i,j  \\  (x,y,i,j)\neq(x_0,y_0,i_0,j_0)}}\lvert z^{(N)}_{x_0,y_0,i_0,j_0}-z^{(N)}_{x,y,i,j}\rvert w_1(z^{(N)}_{x_0,y_0,i_0,j_0})w_2(z^{(N)}_{x,y,i,j})  \\
 & = \limsup_{N\to\infty}\frac{N}{n-1}\frac{1}{N}\log\prod_{\substack{x,y,i,j  \\  (x,y,i,j)\neq(x_0,y_0,i_0,j_0)}}\lvert z^{(N)}_{x_0,y_0,i_0,j_0}-z^{(N)}_{x,y,i,j}\rvert w_1(z^{(N)}_{x_0,y_0,i_0,j_0})w_2(z^{(N)}_{x,y,i,j})  \\
 &+ \frac{N}{n-1}\frac{1}{N}\log w_2(z^{(N)}_{x_0,y_0,i_0,j_0}) - \frac{N}{n-1}\frac{1}{N}\log w_2(z^{(N)}_{x_0,y_0,i_0,j_0}) \\
 & \ge \inf_{z\in\support(\sigma_a)}\int_O\log\lvert z-t\rvert w_1(z)w_2(t)\ed\sigma_a(t)
 -\limsup_{N\to\infty}\frac{N}{n-1}\frac{1}{N}\log w_2(z^{(N)}_{x_0,y_0,i_0,j_0})\\
 &=\inf_{z\in\support(\sigma_a)}\int_O\log\lvert z-t\rvert w_1(z)w_2(t)\ed\sigma_a(t).
\end{split}
\end{equation}
\end{proof}

Since for any $\epsilon>0$ we have $\support(\sigma_a)\subseteq K_\epsilon$ for all sufficiently small $a>0$, \cref{prop:coarsegrainedlowerbound,prop:deltalowerbound} imply that
\begin{equation}
    \delta_{K_\epsilon,w_1,w_2}\ge\inf_{z\in\support(\sigma)}\int_K\log\lvert z-t\rvert w_1(z)w_2(t)\ed\sigma(t).
\end{equation}

\subsection{Upper bound}

We turn to the proof of the first inequality in \eqref{eq:supintbetweendelta}. Following the idea of the proof of \cite[Theorem III.1.3]{saff2013logarithmic} (more precisely, the inequality $\delta_w\le c_w$), we consider the normalized counting measures corresponding to the maximizing sets of points $z_0,\dots,z_n$, and use the weak limit along a subsequence to bound the supremum over the probability measures from below. As the integrand is unbounded, we introduce a cutoff $M>0$ inside the logarithm to ensure convergence of the normalized sums to the integral along the subsequence.

For $M>0$ and $\zeta\in\complexes$ we consider the functions
\begin{equation}
h_{M,\zeta}(t)=\log\max\{M,\lvert\zeta-t\rvert\}.
\end{equation}
\begin{lemma}\label{lem:truncatedloguniformconvergence}
If $\zeta_n\to \zeta$ then $h_{M,\zeta_n}(t)\to h_{M,\zeta}(t)$ holds uniformly in $t\in\complexes$.
\end{lemma}
\begin{proof}
We bound the difference $h_{M,\zeta}(t)-h_{M,\zeta_n}(t)$ as
\begin{equation}
\begin{split}
\log\max\{M,\lvert\zeta-t\rvert\} - \log\max\{M,\lvert\zeta_n-t\rvert\}
 & = \log\frac{\max\{M,\lvert\zeta-t\rvert\}}{\max\{M,\lvert\zeta_n-t\rvert\}}  \\
 & = \log\left(1+\frac{\max\{M,\lvert\zeta-t\rvert\}-\max\{M,\lvert\zeta_n-t\rvert\}}{\max\{M,\lvert\zeta_n-t\rvert\}}\right)  \\
 & \le \log\left(1+\frac{\max\{0,\lvert\zeta-t\rvert-\lvert\zeta_n-t\rvert\}}{\max\{M,\lvert\zeta_n-t\rvert\}}\right)  \\
 & \le \log\left(1+\frac{\lvert\zeta-\zeta_n\rvert}{M}\right),
\end{split}
\end{equation}
in the first inequality using that $x\mapsto\max\{M,x\}$ has Lipschitz constant $1$ and increases, and in the last step using the triangle inequality. Since the roles of $\zeta$ and $\zeta_n$ are symmetric, the same bound holds on the absolute value. The bound is independent of $t$ and vanishes as $\lvert\zeta-\zeta_n\rvert\to 0$ for any fixed $M>0$.
\end{proof}

\begin{proposition}\label{prop:deltaupperbound}
\begin{equation}\label{eq:deltaupperbound}
\delta_{K,w_1,w_2}\le\sup_{\sigma\in\distributions(K)}\inf_{z\in\support(\sigma)}\int_K\log\lvert z-t\rvert w_1(z)w_2(t)\ed\sigma(t)
\end{equation}
\end{proposition}
\begin{proof}
For every $n$ choose an $n+1$-element subset $\mathcal{Z}^{(n)}=\{z^{(n)}_0,z^{(n)}_1,\dots,z^{(n)}_n\}\subseteq K$ maximizing
\begin{equation}
\min_{0\le i\le n}\prod_{\substack{l=0 \\ l\neq i}}^n\left\lvert z^{(n)}_i-z^{(n)}_l\right\rvert w_1(z^{(n)}_i)w_2(z^{(n)}_l).
\end{equation}
Let $\nu_n$ be the normalized counting measure of $\mathcal{Z}^{(n)}$ on $K$ (i.e., $\nu_n(E)=\frac{1}{n+1}\lvert E\cap\mathcal{Z}^{(n)}\rvert$ for every Borel subset $E\subseteq K$).
Choose a subsequence along which
\begin{equation}
\frac{1}{n}\log\min_{0\le i\le n}\prod_{\substack{l=0 \\ l\neq i}}^n\lvert z^{(n)}_i-z^{(n)}_l\rvert w_1(z^{(n)}_i)w_2(z^{(n)}_l)\to\delta_{K,w_1,w_2}.
\end{equation}
With respect to the weak topology, $\distributions(K)$ is compact, therefore we may choose a further subsequence such that $\nu_n\to\nu$ for some $\nu\in\distributions(K)$. From now on, limits in $n$ will always be understood along this subsequence.

Let $z\in\support\nu$ and choose indices $i_n\in\{0,\dots,n\}$ such that $z^{(n)}_{i_n}\to z$. Such indices exist, since for every open neighbourhood $U$ of $z$ we have $0<\nu(U)\le\liminf_{n\to\infty}\nu_n(U)$, which implies $\mathcal{Z}^{(n)}\cap U\neq\emptyset$ for all sufficiently large $n$.

We have for every $M$ that
\begin{equation}
\begin{split}
\delta_{K,w_1,w_2}
 & = \lim_{n\to\infty}\frac{1}{n}\log\min_{0\le i\le n}\prod_{\substack{l=0 \\ l\neq i}}^n\lvert z^{(n)}_i-z^{(n)}_l\rvert w_1(z^{(n)}_i)w_2(z^{(n)}_l)  \\
 & \le \lim_{n\to\infty}\frac{1}{n}\log\prod_{\substack{l=0 \\ l\neq i_n}}^n\lvert z^{(n)}_{i_n}-z^{(n)}_l\rvert w_1(z^{(n)}_{i_n})w_2(z^{(n)}_l)  \\
 & = \lim_{n\to\infty}\frac{1}{n}\sum_{\substack{l=0 \\ l\neq i_n}}^n\left[\log\lvert z^{(n)}_{i_n}-z^{(n)}_l\rvert+\log w_1(z^{(n)}_{i_n})+\log w_2(z^{(n)}_l)\right]  \\
 & \le \lim_{n\to\infty}\Bigg[\log w_1(z^{(n)}_{i_n})+\frac{1}{n}\sum_{\substack{l=0 \\ l\neq i_n}}^n\left(h_{M,z^{(n)}_{i_n}}(z^{(n)}_l)+\log w_2(z^{(n)}_l)\right)\Bigg]  \\
 & = \log w_1(z)+\lim_{n\to\infty}\bigg[\frac{n+1}{n}\int_K\left(h_{M,z^{(n)}_{i_n}}(t)+\log w_2(t)\right)\ed\nu_n(t) \\ &\quad-\frac{1}{n}\log M-\frac{1}{n}\log w_2(z^{(n)}_i)\bigg]  \\
 & \le \log w_1(z)+\lim_{n\to\infty}\bigg[\int_K\left(h_{M,z}(t)+\log w_2(t)\right)\ed\nu_n(t) \\ &\quad+\norm[K]{h_{M,z^{(n)}_{i_n}}(t)-h_{M,z}(t)}\bigg]  \\
 & = \log w_1(z)+\int_K\left(h_{M,z}(t)+\log w_2(t)\right)\ed\nu(t).
\end{split}
\end{equation}
The first equality is due to the choice of $\mathcal{Z}^{(n)}$, in the first inequality we use that the minimum of the products is less than any particular product, the second inequality uses $\log\lvert z^{(n)}_{i_n}-z^{(n)}_l\rvert\le h_{M,z^{(n)}_{i_n}}(z^{(n)}_l)$, the third equality follows by rewriting the sum as an integral with respect to the counting measure and accounting for the omitted term $l=i_n$ (using $h_{M,z^{(n)}_{i_n}}(z^{(n)}_{i_n})=\log M$), in the third inequality we replace $h_{M,z^{(n)}_{i_n}}$ with $h_{M,z}$ in the integral and estimate the error by the supremum norm (on $K$) of the difference times the total measure $\nu_n(K)=1$, and the last equality uses that $\nu_n\to\nu$ weakly, the integrand is continuous, and \cref{lem:truncatedloguniformconvergence}.

This inequality is true for all $M>0$, therefore
\begin{equation}
\begin{split}
\delta_{K,w_1,w_2}
 & \le \lim_{M\to 0}\log w_1(z)+\int_K\left(h_{M,z}(t)+\log w_2(t)\right)\ed\nu(t)  \\
 & = \log w_1(z)+\int_K\left(\log\lvert z-t\rvert+\log w_2(t)\right)\ed\nu(t),
\end{split}
\end{equation}
where the equality follows from the dominated convergence theorem. Finally, the bound holds for all $z\in\support\nu$, therefore
\begin{equation}
\begin{split}
\delta_{K,w_1,w_2}
 & \le \inf_{z\in\support(\nu)}\int_K\log\lvert z-t\rvert w_1(z)w_2(t)\ed\nu(t)  \\
 & \le \sup_{\sigma\in\distributions(K)}\inf_{z\in\support(\sigma)}\int_K\log\lvert z-t\rvert w_1(z)w_2(t)\ed\sigma(t)
\end{split}
\end{equation}
\end{proof}

\subsection{Combining the bounds}

Combining the previously shown bounds, we can now finish the proof of the promised equality between the integral formula and the upper bound on the error exponent arising from the optimized probability bound. 
\begin{corollary}\label{cor:roptwithCompact}
    The optimal achievable error exponent in \eqref{eq:roptdiscretecompact} can be written as
    \begin{equation}\label{eq:optrintegralcompact}
    \begin{split}
        r_\textnormal{opt}&=
        2\inf_{K\subseteq\complexes} \inf_{\sigma\in\mathcal{P}(K)}\sup_{z\in \supp \sigma} 
        \log \norm{A_1(z)\otimes\dots\otimes A_k(z)}+
        e\int\log\frac{\lvert t\rvert}{\lvert z-t\rvert}\ed\sigma(t),
    \end{split}
    \end{equation}
    where the first infimum is taken over the compact subsets of $\complexes$.
\end{corollary}
\begin{proof}
    First we use that the sum of polynomial number of non-negative terms can be estimated asymptotically with the maximum over the terms in the sum, i.e., 
    \begin{equation}
    \begin{split}
        \max_i \prod_{l\neq i} \norm{\sqrt[e]{A_1(z_i)\otimes\dots\otimes A_k(z_i)}}^{2}\frac{\lvert z_l\rvert^2}{\lvert z_i-z_l\rvert^2}&\le
        \sum_{i=1}^{ne+1} \prod_{l\neq i} \norm{\sqrt[e]{A_1(z_i)\otimes\dots\otimes A_k(z_i)}}^{2}\frac{\lvert z_l\rvert^2}{\lvert z_i-z_l\rvert^2}\\&\le
        (ne+1)\max_i \prod_{l\neq i} \norm{\sqrt[e]{A_1(z_i)\otimes\dots\otimes A_k(z_i)}}^{2}\frac{\lvert z_l\rvert^2}{\lvert z_i-z_l\rvert^2}.
    \end{split}
    \end{equation}
    The factor $(ne+1)$ vanishes after we take the logarithm and the limit.
    Now we have to deal with the singularities of the operator. If there is a $z\in \{z_i\}_i\subseteq K$ for which \\$A_1(z)\otimes\dots\otimes A_k(z)$ has a singularity, then the maximum over $i$ becomes infinity. Then it is enough to consider such $K$ sets where $A_1(z)\otimes\dots\otimes A_k(z)$ is continuous. As an immediate consequence of \cref{prop:deltalowerbound,prop:deltaupperbound} with $w_1(z)=\norm{\sqrt[e]{A_1(z)\otimes\dots\otimes A_k(z)}}^{-1}$ and $w_2(t)=\lvert t\rvert^{-1}$ we get
    \begin{equation}
            \begin{split}
        \limsup_{n\to\infty}\frac{1}{n}\log&\inf_{\{z_i\}_i\subseteq K}\sum_{i=1}^{ne+1} \prod_{l\neq i} \norm{\sqrt[e]{A_1(z)\otimes\dots\otimes A_k(z)}}^{2}\frac{\lvert z_l\rvert^2}{\lvert z_i-z_l\rvert^2}\\
        =
        2 &\inf_{\sigma\in\mathcal{P}(K)}\sup_{z\in \supp \sigma} 
        \log \norm{A_1(z)\otimes\dots\otimes A_k(z)}+
        e\int\log\frac{\lvert t\rvert}{\lvert z-t\rvert}\ed\sigma(t).
    \end{split}
    \end{equation}
\end{proof}

The next proposition concludes our final expression for the optimal error exponent.
\begin{proposition}\label{prop:rFinalForm}
Assume that $\norm{A_1(z)\otimes\dots\otimes A_k(z)}\to\infty$ whenever $\lvert z\rvert\to 0$ and $\norm{A_1(z)\otimes\dots\otimes A_k(z)}$ is nowhere zero (see \cref{rem:degenInterestingcases}).
Then
\begin{equation}
    \begin{split}
    &\inf_{K\subseteq\complexes}\inf_{\sigma\in\mathcal{P}(K)}\sup_{z\in \supp \sigma} \left[
        \log \norm{A_1(z)\otimes\dots\otimes A_k(z)}+
        e\int\log\frac{\lvert t\rvert}{\lvert z-t\rvert}\ed\sigma(t)\right]
        \\=
        &\inf_{\sigma\in\mathcal{P}(\complexes)}\sup_{z\in \supp \sigma} \left[
        \log \norm{A_1(z)\otimes\dots\otimes A_k(z)}+
        e\int_\complexes
        \log\frac{\lvert t\rvert}{\lvert z-t\rvert}\ed\sigma(t)\right],
    \end{split}
\end{equation}
where the first infimum is taken over the compact subsets of $\complexes$.

\end{proposition}
\begin{proof}
On one hand we trivially have that the right hand side is less than or equal to the left hand side. For the converse let $\sigma\in\distributions(\complexes)$ and suppose that
\begin{equation}
S(\sigma)\coloneqq\sup_{z\in\support(\sigma)}\left[\log A(z)+e\int_\complexes\log\frac{\lvert t\rvert}{\lvert z-t\rvert}\ed\sigma(t)\right]<\infty,
\end{equation}
where $A(z) \coloneqq \norm{A_1(z)\otimes\dots\otimes A_k(z)}$.

For the optimal $\sigma$ measure $0\not\in\support(\sigma)$ since the second term vanishes at $z=0$ and the first term is $\infty$, so it is enough to consider $\sigma$ measures satisfying this property.
In the following we show that for any such $\sigma$ and $\epsilon>0$ there exists a probability measure $\sigma_R$ with compact support such that $S(\sigma)+\epsilon\ge S(\sigma_R)$.

For all sufficiently large $R$ we have $\sigma(\closedball{R}{0})>0$, therefore we can define the conditional distribution
\begin{equation}
\sigma_R(E)=\frac{\sigma(E\cap\closedball{R}{0})}{\sigma(\closedball{R}{0})},
\end{equation}
which is supported on a compact subset of $\complexes\setminus\{0\}$. 
For $\epsilon>0$ choose $z_R\in\support(\sigma_R)$ such that
\begin{equation}
S(\sigma_R)-\epsilon\le \log A(z_R)+e\int_\complexes\log\frac{\lvert t\rvert}{\lvert z_R-t\rvert}\ed\sigma_R(t).
\end{equation}
If $R\le\lvert t\rvert$, then
\begin{equation}
\frac{\lvert t\rvert}{\lvert z_R-t\rvert}
 \ge \frac{\lvert t\rvert}{\lvert z_R\rvert+\lvert t\rvert}
 =  \frac{1}{\lvert \frac{z_R}{t}\rvert+1}
 \ge \frac{1}{2},
\end{equation}
since $\lvert z_R\rvert\le R$.
This implies
\begin{equation}
\begin{split}
S(\sigma)
 & \ge \log A(z_R)+e\int_\complexes\log\frac{\lvert t\rvert}{\lvert z_R-t\rvert}\ed\sigma(t)  \\
 & = \log A(z_R)+e\int_{\closedball{R}{0}}\log\frac{\lvert t\rvert}{\lvert z_R-t\rvert}\ed\sigma(t)+e\int_{\complexes\setminus\closedball{R}{0}}\log\frac{\lvert t\rvert}{\lvert z_R-t\rvert}\ed\sigma(t)  \\
 & \ge \log A(z_R)+e\sigma(\closedball{R}{0})\int_{\closedball{R}{0}}\log\frac{\lvert t\rvert}{\lvert z_R-t\rvert}\ed\sigma_R(t)-e(1-\sigma(\closedball{R}{0}))  \\
 & = \sigma(\closedball{R}{0})\left[\log A(z_R)+e\int_{\closedball{R}{0}}\log\frac{\lvert t\rvert}{\lvert z_R-t\rvert}\ed\sigma_R(t)\right]+(1-\sigma(\closedball{R}{0}))\left(\log A(z_R)-e\right)  \\
 & \ge \sigma(\closedball{R}{0})(S(\sigma_R)-\epsilon)+(1-\sigma(\closedball{R}{0}))(\log A(z_R)-e),
\end{split}
\end{equation}
therefore
\begin{equation}
\frac{1}{\sigma(\closedball{R}{0})}S(\sigma)-\frac{1-\sigma(\closedball{R}{0})}{\sigma(\closedball{R}{0})}(\log A(z_R)-e)+\epsilon\ge S(\sigma_R).
\end{equation}
Since $\lim_{R\to\infty}\sigma(\closedball{R}{0})=1$ and $\log A(z)$ is bounded from below, the limit of the left hand side is $S(\sigma)+\epsilon$. This holds for all $\epsilon>0$, therefore the converse inequality also holds.
\end{proof}

\begin{proof}[Proof of {\cref{thm:main}}]
Note that if $A_1(z)\otimes\dots\otimes A_k(z)$ contains both positive and negative powers of $z$ and it is nowhere zero, then the conditions of \cref{prop:rFinalForm} are satisfied for $\norm{A_1(z)\otimes\dots\otimes A_k(z)}$.
Then by the previous proposition the optimal error exponent takes its final form:
    \begin{equation}\label{eq:optrintegralFinal}
    \begin{split}
        r_\textnormal{opt}&=
        2 \inf_{\sigma\in\mathcal{P}(\complexes)}\sup_{z\in \supp \sigma} \left[
        \log \norm{A_1(z)\otimes\dots\otimes A_k(z)}+
        e\int_\complexes
        \log\frac{\lvert t\rvert}{\lvert z-t\rvert}\ed\sigma(t)\right].
    \end{split}
    \end{equation}
\end{proof}

\section{Bounds and special cases}\label{sec:bounds}

Although \eqref{eq:optrintegralFinal} is a single-letter upper bound on the optimal strong converse exponent, it involves an infimum over the set of probability measures, and in this generality we do not see a way to further simplify the expression. In the following we derive a simple lower bound on the infimum which, as we will see, is tight in the important special case when $\log \norm{A_1(z)\otimes\dots\otimes A_k(z)}$ is a centrally symmetric function of $z$.
\begin{lemma}\label{lem:Aconst}
\begin{equation}
        \inf_{\sigma\in\mathcal{P}(\complexes)}\sup_{z\in \supp \sigma} 
        \int\log\lvert t\rvert-\log\lvert z-t\rvert \ed\sigma(t)=0.
\end{equation}
\end{lemma}
\begin{proof}
First note that by \cref{lem:harmonic} and the fact that $\int\log\lvert t\rvert-\log\lvert z-t\rvert \ed\sigma(t)\to-\infty$ if $\lvert z\rvert\to\infty$, the supremum in $z$ over $\complexes$ is the same as the supremum over $\supp\sigma$. Therefore we have
  \begin{equation}\label{eq:integralLowerByZero}
  \begin{split}
        \sup_{z\in \supp \sigma}\int\log\lvert t\rvert-\log\lvert z-t\rvert\ed\sigma(t)
        & = \sup_{z\in \complexes}\int\log\lvert t\rvert-\log\lvert z-t\rvert\ed\sigma(t) \\
        & \ge \int\log\lvert t\rvert-\log\lvert t\rvert\ed\sigma(t)
          =0.      
  \end{split}
    \end{equation}

For the converse we choose the uniform probability measure over the unit circle $\unitcircle$ (any circle would do) centered at the origin. 
\begin{equation}
\begin{split}
\inf_K \inf_{\sigma\in\mathcal{P}(K)}\sup_{z\in \supp \sigma} 
        \int\log\lvert t\rvert-\log\lvert z-t\rvert \ed\sigma(t)
        &\le\sup_{z\in \supp \unitcircle}\int\log\lvert t\rvert-\log\lvert z-t\rvert \ed\sigma_\unitcircle(t)  \\
    &=\sup_{z\in \supp \unitcircle}\left[\log 1 -\log 1\right]
     =0.
\end{split}
\end{equation}
\end{proof}

\begin{corollary}\label{cor:optlowerbound}
    The optimal value of the error exponent in \cref{cor:roptwithCompact} can be lower bounded as
    \begin{equation}
        r_\textnormal{opt}\ge 2\log \min_z \norm{A_1(z)\otimes\dots\otimes A_k(z)}.
    \end{equation}
\end{corollary}
\begin{proof}
  The right hand side of \eqref{eq:optrintegralFinal} can be lower bounded by using $\norm{A_1(z)\otimes\dots\otimes A_k(z)}\ge\min_z \norm{A_1(z)\otimes\dots\otimes A_k(z)}$, then the rest is the immediate consequence of \cref{lem:Aconst}.
\end{proof}

\begin{proposition}\label{prop:symmetricA}
    Let $A_1(z)\otimes\dots\otimes A_k(z)$ be centrally symmetric, i.e., $A_1(z)\otimes\dots\otimes A_k(z)=A_1(ze^{i\varphi})\otimes\dots\otimes A_k(ze^{i\varphi})$ for any $z\in\complexes$ and $\varphi\in\reals$. Then the optimal error exponent $r_\textnormal{opt}$ achieves its lower bound in $\cref{cor:optlowerbound}$. 
\end{proposition}
\begin{proof}
    By choosing $\sigma$ to be the uniform measure on the circle with radius $r$ around the origin, we get
    \begin{equation}
    \begin{split}
        \int\log\lvert t\rvert-\log\lvert z-t\rvert \ed\sigma(t)
        & = \log r - \sup_{z\in\supp\sigma}\int\log\lvert z-re^{i\varphi}\rvert \ed\varphi  \\
        & = \log r - \sup_{z\in\supp\sigma}\log r =0.
    \end{split}
    \end{equation}

On the other hand, $\norm{A_1(z)\otimes\dots\otimes A_k(z)}$ only depends on $\lvert z\rvert$ so one can chose $r$ such that $\min_z \norm{A_1(z)\otimes\dots\otimes A_k(z)}= 
\norm{A_1(re^{i\varphi})\otimes\dots\otimes A_k(re^{i\varphi})}$.
\end{proof}

\begin{example}\label{ex:GHZtoW0}
Let $W=\frac{1}{\sqrt{k}}\left(\ket{100\dots 0} + \ket{010\dots 0} +\dots+ \ket{0\dots 01}\right)$ be the $k$-party W state and write
    \begin{multline}
    \frac{1}{z}\sqrt{\frac{2}{k}}
\left( \begin{matrix}
1 & -1 \\
z & 0
\end{matrix} \right)
\otimes \cdots \otimes
\left( \begin{matrix}
1 & -1 \\
z & 0
\end{matrix} \right)
\frac{1}{\sqrt{2}}(\ket{00\ldots0} - \ket{11\ldots1})  \\ = \frac{1}{\sqrt{k}}(\ket{10\dots}+\ket{010\dots}+\dots+\ket{0\dots 01}) + O(z).
\end{multline}
The state $\frac{1}{\sqrt{2}}(\ket{00\dots 0} -\ket{11\dots 1})$ is LU-equivalent to the generalized GHZ state, therefore this equation can be seen as a degeneration between GHZ and W. The norm of the operator on the left hand side is 
\begin{equation}
    \norm{A_1(z)\otimes\dots\otimes A_k(z)}= \sqrt{\frac{2}{k}}\frac{(4+\lvert z\rvert^2)^{\frac{k}{2}}}{\lvert z\rvert}.
\end{equation}
This is clearly centrally symmetric so \cref{prop:symmetricA} can be applied, and we get the error exponent
\begin{equation}
    r_{\GHZ\to\W}=\inf_{d>0}2\log \frac{(4+d^2)^{\frac{k}{2}}}{d}
    +\log\frac{2}{k}
    = 2(k - 1)+  k\log k - (k-1)\log(k - 1)
    +\log\frac{2}{k}
\end{equation}
where the minimum is achieved by $d=\sqrt{\frac{4}{k-1}}$. For $k=3$ this gives $r_{\GHZ\to\W}\approx 6.17$
\end{example}

Although \cref{ex:GHZtoW0} gives the best possible error exponent achievable by this particular degeneration, it is in general possible to find different degenerations between the same pair of states, which do not necessarily lead to the same error exponent. We illustrate this with another degeneration from GHZ to the W state, which is an example of a combinatorial degeneration. First we recall the definition and specialize \cref{prop:symmetricA} to combinatorial degenerations.
\begin{definition}\label{def:combinatorialDegen}
Let $I_1,\dots,I_k$ be finite index sets and let $\Phi \subseteq \Psi \subseteq I_1 \times \cdots \times I_k$. We say that $\Phi$ is a \emph{combinatorial degeneration} of $\Psi$ if there are maps $u_j :\, I_j \to \integers$ such that for all $\alpha \in I_1 \times \cdots \times I_k$, if $\alpha \in \Psi \setminus \Phi$, then $\sum_{j=1}^k u_j(\alpha_j) > 0$, and if $\alpha \in \Phi$, then $\sum_{j=1}^k u_j(\alpha_j) = 0$. 
\end{definition}
For further details on combinatorial degeneration see \cite{christandl2023universal,burgisser2013algebraic}.
Let $\{\ket{i}\}_{i\in I_j}$ be bases on the local Hilbert spaces $\mathcal{H}_j$. Then acting with the linear maps
\begin{equation}
    A_{j}\coloneqq
    \sum_{i\in I_j} z^{u_j(i)}\vectorstate{i},
\end{equation}
where $j\in [k]$, on a state with support in $\Phi$ we get a degeneration to a state with support in $\Psi$. 

\begin{proposition}\label{prop:combinatorialDegen}
Let $I_1,\dots,I_k$ be finite index sets and $\Phi \subseteq \Psi \subseteq I_1 \times \cdots \times I_k$, and suppose that $\Phi$ is a combinatorial degeneration of $\Psi$. Let
\begin{equation}
\psi=\sum_{(i_1,\dots,i_k)\in\Psi}\psi_{i_1,\dots,i_k}\ket{i_1}\otimes\dots\otimes\ket{i_k}
\end{equation}
be a unit vector, $q=\sum_{(i_1,\dots,i_k)\in\Phi}\lvert\psi_{i_1,\dots,i_k}\rvert^2$, and
\begin{equation}
\varphi=\frac{1}{\sqrt{q}}\sum_{(i_1,\dots,i_k)\in\Phi}\psi_{i_1,\dots,i_k}\ket{i_1}\otimes\dots\otimes\ket{i_k}.
\end{equation}
Then $\psi$ can be asymptotically transformed into $\varphi$ with rate $1$ and strong converse exponent at most $r=-\log q$.
\end{proposition}
\begin{proof}
Let $u_j:I_j\to\integers$ as in \cref{def:combinatorialDegen}. Consider the maps
\begin{equation}
A_j(z)\coloneqq q^{-\frac{1}{2k}}\sum_{i\in I_j} z^{u_j(i)}\vectorstate{i}.
\end{equation}
Then
\begin{equation}
\begin{split}
(A_1(z)\otimes\dots\otimes A_k(z))\psi
 & = \frac{1}{\sqrt{q}}\sum_{(i_1,\dots,i_k)\in\Psi}\psi_{i_1,\dots,i_k}z^{u_1(i_1)+\dots+u_k(i_k)}\ket{i_1}\otimes\dots\otimes\ket{i_k}  \\
 & = \varphi+O(z).
\end{split}
\end{equation}
The norm of the product map is
\begin{equation}
\norm{A_1(z)\otimes\dots\otimes A_k(z)}=\begin{cases}
    \frac{1}{\sqrt{q}}\lvert z\rvert^{u_\textnormal{max}} & \text{if $\lvert z\rvert\ge 1$}  \\
    \frac{1}{\sqrt{q}}\lvert z\rvert^{u_\textnormal{min}} & \text{otherwise},
\end{cases}
\end{equation}
where $u_\textnormal{max}=\sum_{j=1}^k\max_{i\in I_j}u_j(i)$ and $u_\textnormal{min}=\sum_{j=1}^k\min_{i\in I_j}u_j(i)$. In particular, $\norm{A_1(z)\otimes\dots\otimes A_k(z)}$ is centrally symmetric and its minimum is $\frac{1}{\sqrt{q}}$, therefore the claim follows from \cref{prop:symmetricA}.
\end{proof}

The following example shows a combinatorial degeneration (see \cref{prop:combinatorialDegen}) between the GHZ and the W state.  
\begin{example}\label{ex:GHZtoW}
    Let
    \begin{equation}
    \GHZ=\frac{1}{2}\left(
    \ket{+++}+\ket{--+}+\ket{+--}+\ket{-+-}
    \right),    
    \end{equation}
    be the $\GHZ$ state expressed in the basis $\{\ket{+},\ket{-} \}$, and 
    \begin{equation}
        \W=\frac{1}{\sqrt{3}}(\ket{+--}+\ket{-+-}+\ket{--+}),
    \end{equation}
    a $\W$ state with $\{\ket{+},\ket{-} \}$ being its canonical basis.
Let
\begin{equation}
    A_j(z)=\sqrt[3]{\frac{2}{\sqrt{3}}}
    \sum_{i\in\{+,-\}} z^{u(i)}\vectorstate{i}
\end{equation}
 for any $j\in [k]$ and
\begin{equation}
     A_1(z)\otimes A_2(z)\otimes A_3(z)=\frac{2}{\sqrt{3}} \sum_{i_1,i_2,i_3 \in \{+,-\}}
     z^{\sum_{j=1}^3 u(i_j)}\vectorstate{i_1,i_2,i_3}
     ,
\end{equation}
where $u(+)=2$ and $u(-)=-1$.
Note that any basis vector in the support of the W state is also in the support of the GHZ state. Also note that for the basis vectors in the support of the GHZ state we have $\sum_{j=1}^3 u_j(i)\ge 0$ and there is an equality iff the basis vector is also in the support of the W state. 
Acting on the GHZ state then results in
    \begin{equation}
    \begin{split}
        A_1(z)\otimes A_2(z)\otimes A_3(z)\GHZ&=\frac{2}{\sqrt{3}}
        \sum_{i_1,i_2,i_3 \in \{+,-\}}
     z^{\sum_{i=j}^3 u(i_j)}\vectorstate{i_1,i_2,i_3}
     \GHZ\\
     &=\frac{1}{\sqrt{3}}\left(\ket{+--}+\ket{-+-}+\ket{--+}+z^6\ket{+++}
    \right)
    \end{split}
    \end{equation}
The norm of the operator on the left hand side is 
\begin{equation}
    \norm{A_1(z)\otimes A_2(z)\otimes A_3(z)}= 
    \begin{cases}
    \frac{2}{\sqrt{3}} \lvert z\rvert^{6}\, &\text{if}\quad\lvert z \rvert \ge 1 \\
    \frac{2}{\sqrt{3}} \lvert z\rvert^{-3}\,&\text{if}\quad\lvert z \rvert < 1 .
    \end{cases}
\end{equation}
This is again centrally symmetric so we can apply \cref{prop:symmetricA}, and we get the error exponent 
\begin{equation}
    r_{\GHZ\to\W}=2\log\frac{2}{\sqrt{3}}\approx 0.415.
\end{equation}
\end{example}

\begin{example}\label{ex:graphtensors}
Consider the transformations studied in \cite{vrana2017entanglement}. The initial state $\psi$ is a product of $\GHZ$ states shared among subsets of the $k$ parties, encoded in a hypergraph $H$ (possibly with parallel hyperedges). The vertex set is $[k]$, and a hyperedge $E$ incident with a set of vertices corresponds to a $\GHZ$ state on that subset (an $\EPR$ pair, when the hyperedge has size $2$). The target state is the $\GHZ$ state on \emph{all} subsystems. The optimal asymptotic SLOCC rate is equal to the edge-connectivity $\lambda(H)$ of the hypergraph, defined as the largest number $l$ such that after removing any subset of at most $l-1$ hyperedges the hypergraph remains connected. As a concrete example, the complete graph $K_3$ on $3$ vertices corresponds to a triple of $\EPR$ pairs arranged as $\EPR_{AB}\otimes\EPR_{AC}\otimes\EPR_{BC}$, with edge-connectivity $\lambda(K_3)=2$.

The asymptotic SLOCC transformation with this rate is in general not possible for any finite number of copies (even allowing finite-copy degenerations \cite[Theorem 6.1]{kopparty2023geometric}). In \cite{vrana2017entanglement}, a sequence of combinatorial degenerations (in the standard basis) is found, from $\psi^{\otimes n}$ to a $\GHZ$ state with at least $C2^{n\lambda(H)}$ levels for some constant $C>0$. Since the squared coefficients in the standard basis are distributed uniformly on $2^{n\lvert E(H)\rvert}$ elements (where $E(H)$ is the set of hyperedges of $H$), and the combinatorial degeneration is to a subset of at least $C2^{n\lambda(H)}$ elements, such a degeneration implies via \cref{prop:combinatorialDegen} that a transformation from $\psi$ to $\GHZ$ is possible at rate
\begin{equation}
R\ge\frac{1}{n}\log\left(C2^{n\lambda(H)}\right)
\end{equation}
and error exponent
\begin{equation}
r\le-\frac{1}{n}\log\frac{C2^{n\lambda(H)}}{2^{n\lvert E(H)\rvert}}.
\end{equation}
Letting $n\to\infty$, we obtain a transformation at the optimal asymptotic SLOCC rate $R=\lambda(H)$, with a strong converse exponent $\lvert E(H)\rvert-\lambda(H)$.
\end{example}

As mentioned above, the formula for the optimal error exponent is reminiscent of the equality between the weighted transfinite diameter and the weighted capacity. While it does not seem to be possible to reduce our result to that equality, we can derive a lower bound on the error exponent in terms of a weighted capacity, either by estimating directly the integral formula, or by symmetrizing the two weights $w_1,w_2$ in \eqref{eq:deltaKw1w2def}.
\begin{proposition}
    The optimal value of the error exponent in \eqref{eq:optrintegralcompact} is lower bounded by the weighted capacity 
    \begin{equation}
    \begin{split}
        r_\textnormal{opt}&=2\inf_{\sigma\in\mathcal{P}(\complexes)}\sup_{z\in \supp \sigma} 
        \log \norm{A_1(z)\otimes\dots\otimes A_k(z)}+
        e\int_\complexes \log\frac{\lvert t\rvert}{\lvert z-t\rvert}\ed\sigma(t)\\
        &\ge
        2\inf_\sigma e\iint \log \frac{1}{w(z)w(t)\lvert z-t\rvert}\ed\sigma(z)\ed\sigma(t),
    \end{split}
    \end{equation}
with weights $w(z)=1/\sqrt{\norm{A_1(z)\otimes\dots\otimes A_k(z)}^{\frac{1}{e}}\lvert z\rvert}$.
\end{proposition}
\begin{proof}
This is done by lower bounding the supremum in $z$ by the integral in $z$ over the probability measure $\sigma$. Then by algebraic manipulations we get 
\begin{equation}\label{eq:integrallowerboundbycapacity}
\begin{split}
    &2\inf_{\sigma\in\mathcal{P}(\complexes)}\sup_{z\in \supp \sigma} 
        \log \norm{A_1(z)\otimes\dots\otimes A_k(z)}+
        e\int_\complexes \log\frac{\lvert t\rvert}{\lvert z-t\rvert}\ed\sigma(t)\ge\\
        &2e\left(\inf_{\sigma\in\mathcal{P}(\complexes)}\iint 
        \log\frac{1}{\lvert z-t\rvert}\ed\sigma(t)\ed\sigma(z)
        +\frac{1}{2}\int\log\sqrt[e]{\norm{A_1(z)\otimes\dots\otimes A_k(z)}}\ed\sigma(z)\right. \\
        &\left.+\frac{1}{2}\int\log\lvert z\rvert\ed\sigma(z)
        +\frac{1}{2}\int\log\sqrt[e]{\norm{A_1(t)\otimes\dots\otimes A_k(t)}}\ed\sigma(t)
        +\frac{1}{2}\int\log\lvert t\rvert\ed\sigma(t)
        \right)=\\
        &2e\left(\inf_{\sigma\in\mathcal{P}(\complexes)}\iint 
        \log\frac{1}{w(z)w(t)\lvert z-t\rvert}\ed\sigma(t)\ed\sigma(z)
        \right).
\end{split}
\end{equation}
\end{proof}

Alternatively, one can show the same inequality by using the inequality between the arithmetic and geometric means:
\begin{multline}
\sum_{i=1}^{ne+1} \prod_{k\neq i} \sqrt[e]{\norm{A_1(z_i)\otimes\dots\otimes A_k(z_i)}^{2}}\frac{\lvert z_k\rvert^2}{\lvert z_i-z_k\rvert^2}  \\
\begin{aligned}
& > \left(\prod^{ne}_{j=0} \prod_{k\neq j} \sqrt[e]{\norm{A_1(z_j)\otimes\dots\otimes A_k(z_j)}^2} \frac{\lvert z_k\rvert^2 }{\lvert z_j-z_k\rvert^2}\right)^{\frac{1}{ne+1}}  \\
& = \left(\prod^{ne}_{j=0} \prod_{k\neq j}  \frac{1}{w^2(z_j)w^2(z_k)\lvert z_j-z_k\rvert^2}\right)^{\frac{1}{ne+1}},
\end{aligned}
\end{multline}
where the equality is the result of a similar symmetrizing we have done with the integrals in \eqref{eq:integrallowerboundbycapacity}.
Then the rest is done by the fact that the weighted transfinite diameter equals with the weighted capacity (see \cref{subsec:transfiniteNpot}).

We have seen in the previous section that the optimal achievable error exponent is given by \cref{prop:symmetricA} when $\norm{A_1(z)\otimes\dots\otimes A_k(z)}$ is centrally symmetric. Unfortunately, in the absence of such symmetry, the problem becomes much more complicated. Still, in those cases we can formulate a calculable upper bound for the optimal error exponent.

\begin{proposition}\label{prop:logcircle}
    For any $R>0$ the optimal error exponent admits the upper bound
    \begin{equation}
    \begin{split}
         r_\textnormal{opt}&= 2\inf_{\sigma\in\mathcal{P}(\complexes)}\sup_{z\in \supp \sigma} 
        \log\norm{A_1(z)\otimes\dots\otimes A_k(z)}+
        e\int_\complexes \log\frac{\lvert t\rvert}{\lvert z-t\rvert}\ed\sigma(t)\\&\le
        2\int_0^{2\pi} \log\norm{A_1(Re^{i\varphi})\otimes\dots\otimes A_k(Re^{i\varphi})}\ed\varphi.
    \end{split}
    \end{equation}
\end{proposition}
\begin{proof}
    Let $\sigma$ be the probability measure on the circle $C_R$ with radius $R$ defined by a density function $\rho(\varphi)$ and let $\varphi_z$ and $\varphi_t$ be the angles corresponding to $z$ and $t$ respectively. For $\epsilon>0$ we define $A_\epsilon(z)\coloneqq\max\{\norm{A_1(z)\otimes\dots\otimes A_k(z)},\epsilon\}$. 
    Using $A_\epsilon(z)$ instead of $\norm{A_1(z)\otimes\dots\otimes A_k(z)}$ we get an upper bound for the error exponent. With this the integrand becomes continuous on $C_R$, therefore it can be uniformly approximated by its Fourier series, in other words for any $\epsilon>0$ we can choose $\lvert n\rvert$ large enough so that the partial sum from $-n$ to $n$ of the Fourier series of $\log \norm{A_1(Re^{i\varphi})\otimes\dots\otimes A_k(Re^{i\varphi})}$ differs from the function by $\epsilon$ at most for any $z\in C_R$. Our aim is to eliminate any non-constant Fourier terms from this series with the integral term by a suitable choice of measure.
    Now we rewrite the integral using the law of cosines:
    \begin{equation}
    \begin{split}
        \int_{C_R} \log\frac{\lvert t\rvert}{\lvert z-t\rvert}\ed\sigma(t)&=
        \int_{C_R} \log\frac{R}{\sqrt{2}R\sqrt{1-\cos(\varphi_z-\varphi_t)}}\ed\sigma(t)\\&=
        \int_0^{2\pi} \log\frac{1}{\sqrt{2(1-\cos(\varphi_z-\varphi_t))}}\rho(\varphi_t)\ed\varphi,
    \end{split}
    \end{equation}
    where $\varphi$ is the angle between $z$ and $t$ on the circle. Note that here we have a convolution of two functions which in Fourier space converts to a multiplication. To apply this fact, first we calculate the Fourier coefficients of the kernel function:
    \begin{equation}
        c_0=\int_0^{2\pi} \log\frac{1}{\sqrt{2(1-\cos(\varphi))}}\ed\varphi=
        0
    \end{equation}
    and 
    \begin{equation}
        c_m=\int_0^{2\pi} \log\frac{1}{\sqrt{2(1-\cos(\varphi))}}e^{-im\varphi}\ed\varphi=
        \frac{1}{2\lvert m\rvert},
    \end{equation}
    for $m\neq 0$, $-n\le m\le n$.
    Let $\hat{\rho}_m$ be the Fourier polynomial of the density function $\rho$ and $\hat{A}_m$ be the Fourier polynomial of $-\log\norm{A_1(z)\otimes\dots\otimes A_k(z)}$, both with the cutoff $-n\le m\le n$. We want to achieve $\hat{A}_m$ by the given integral transformation. By the convolution theorem we have $\hat{A}_0=0$ and $\hat{A}_m=\frac{e\hat{\rho}_m}{2\lvert m\rvert}$.
    This means that by choosing 
    \begin{equation}
    \hat{\rho}_m=\frac{2\lvert m\rvert\hat{A}_m}{e}
    \end{equation}
for any $m\neq 0$, the non-constant terms from the Fourier series of $\log\norm{A_1(z)\otimes\dots\otimes A_k(z)}$ are eliminated. The only problem with this is that it is not guaranteed that $\rho$ will be a density function of a probability measure. To this end we chose $\hat{\rho}_0=1$ to ensure normality, and $\sum_{n\neq 0}\lvert\hat{\rho}_m\rvert=\sum_{m\neq 0}\frac{2\lvert m\hat{g}_m\rvert}{e}\le 1$ to ensure non-negativity. The second condition can be fulfilled only by choosing $e$ large enough, but that choice can be made since the polynomial on the right hand side of \eqref{eq:degenpolynomial} can always be expanded by larger power terms with zero coefficients, which enlarges its error degree $e$.

At the end, for any $\epsilon$ we can choose $n$ and $e$ large enough so that there exists a probability distribution with density function $\rho$ such that
\begin{multline}
        \inf_{\sigma\in\mathcal{P}(\complexes)}\sup_{\varphi_z} 
        \log A_\epsilon(Re^{i\varphi_z})+
        e\int_\complexes \log\frac{\lvert t\rvert}{\lvert z-t\rvert}\ed\sigma(t)  \\
\begin{aligned}
        & \le \sup_{\varphi_z} \log A_\epsilon(Re^{i\varphi_z})-\left(\log A_\epsilon(Re^{i\varphi_z})-\int_0^{2\pi} \log A_\epsilon(Re^{i\varphi})\ed\varphi\right)+\epsilon  \\
        &=\int_0^{2\pi} \log A_\epsilon(Re^{i\varphi})\ed\varphi+\epsilon
    \end{aligned}
    \end{multline}
The only thing left to show is the fact that the average of $A_\epsilon(z)$ in the limit $\epsilon\to 0$ is the same as the average of $\norm{A_1(z)\otimes\dots\otimes A_k(z)}$. When a polynomial described in \eqref{eq:degenpolynomial} takes the value zero at some point $z_0$ then let $\alpha\ge 1$ be the largest integer for which $(z-z_0)^\alpha$ can be factored out from the polynomial. Then by factoring this out from $\norm{A_1(z)\otimes\dots\otimes A_k(z)}$ as $\lvert(z-z_0) \rvert^\alpha$ the remaining function lacks singularity at $z_0$. By taking the logarithm we get a singular term in a form $\alpha\log\lvert z-z_0 \rvert$, which in terms of angles translates to $\lvert z-z_0 \rvert=R \sqrt{2(1-\cos\varphi)}$ where $\varphi$ is the angle between $z_0$ and $z$. But
\begin{equation}
    \int_{-\epsilon}^\epsilon \log R \sqrt{2(1-\cos\varphi)}\ed\varphi=2\epsilon\log R +
    2\epsilon\log \epsilon -2\epsilon,
\end{equation}
which goes to $0$ as $\epsilon\to 0$ so this singularity does not contribute to the integral. Then we have
\begin{equation}
    \lim_{\epsilon\to 0}\int_0^{2\pi} \log A_\epsilon(Re^{i\varphi})\ed\varphi=\int_0^{2\pi} \log
    \norm{A_1(Re^{i\varphi})\otimes\dots\otimes A_k(Re^{i\varphi})}
    \ed\varphi.
\end{equation}
\end{proof}

\begin{remark}
    Note that by operating under the assumptions of \cref{rem:degenInterestingcases} in \cref{prop:logcircle} we do not have to deal with the singularities of $\log\norm{A_1(z)\otimes\dots\otimes A_k(z)}$, and the proof can be simplified by immediately considering the Fourier series of $\log\norm{A_1(z)\otimes\dots\otimes A_k(z)}$.
\end{remark}

\begin{corollary}\label{cor:logcircle}
    Let $C_R$ be a circle around the origin with radius $R$. Then the error exponent
    \begin{equation}
        r=2\inf_{R>0}\int_0^{2\pi} \log \norm{A_1(Re^{i\varphi})\otimes\dots\otimes A_k(Re^{i\varphi})}\ed\varphi,
    \end{equation}
    is achievable.
\end{corollary}
\begin{proof}
    Immediate consequence of \cref{prop:logcircle}.
\end{proof}

\section{Error exponents for transformation rates below one}\label{sec:errorexptrfratelessone}

In this section, we study the achievable error exponents for protocols that utilize degeneration when the transformation rate $R < 1 $. By relaxing the rate, we can aim for smaller error exponents than those observed in the  $R = 1$ case. Indeed, this reduction is achievable by discarding $(1-R)n$ copies of the initial state and applying the protocol with $R = 1$ to the remaining copies, resulting in a success probability of $2^{-rRn+o(n)}$.

Similarly, for protocols characterized by rate-exponent pairs  $(R_1, r_1)$ and $(R_2, r_2)$, time-sharing can be employed. By dividing the initial state into groups of $pn$ and $(1-p)n$ copies, and applying the respective protocols to these subsets, the resulting rate is $pR_1 + (1-p)R_2$ with a success probability of $2^{-r_1pn+o(n)} \cdot 2^{-r_2(1-p)n+o(n)} = 2^{-\left(r_1p + r_2(1-p)\right)n+o(n)}$. This demonstrates the convexity of the achievable region in the $(R, r)$ plane.

In the following we construct a protocol which gives a better error exponent for $R<1$ than the previously introduced time-sharing protocol. 
Given $n\in\naturals$ and complex numbers $z_1,\dots,z_{t}$, we form the operators $\mathbb{C}^{t}\otimes\mathcal{H}_j^{\otimes n}\to\mathbb{C}^2\otimes\mathbb{C}^{t}\otimes\mathcal{H}_j^{\otimes n}$
\begin{align}
K^{(1)}_{j,m} & = \sum_{i=1}^{t}\ket{1}\otimes\ketbra{i}{i}_j\otimes\underbrace{I\otimes\dots\otimes I}_{m-1}\otimes \frac{A_j(z_i)}{\norm{A_j(z_i)}}\otimes\underbrace{I\otimes\dots\otimes I}_{n-m}  \\
K^{(0)}_{j,m} & = \sum_{i=1}^{t}\ket{0}\otimes\ketbra{i}{i}_j\otimes\underbrace{I\otimes\dots\otimes I}_{m-1}\otimes\sqrt{I-\frac{A_j(z_i)^*A_j(z_i)}{\norm{A_j(z_i)}^2}}\otimes\underbrace{I\otimes\dots\otimes I}_{n-m}.
\end{align}
For different $j$ or $m$ these commute (since the only ``shared'' subsystem is $\complexes^{t}$ where they are all diagonal) and $(K^{(0)}_{j,m})^*K^{(0)}_{j,m}+(K^{(1)}_{j,m})^*K^{(1)}_{j,m}\le I$, therefore for every $j$ the products
\begin{equation}
K^x_j=\prod_{m=1}^n K^{(x_m)}_{j,m}
\end{equation}
for $x\in\{0,1\}^n$ are the Kraus operators of a local channel $T_j$ from $\boundeds(\mathbb{C}^{t}\otimes\mathcal{H}_j^{\otimes n})\to\boundeds((\mathbb{C}^{2})^{\otimes n}\otimes\mathbb{C}^{t}\otimes\mathcal{H}_j^{\otimes n})$. We apply these to $n$ copies of the initial state and a GHZ state and let each party measure all the $n$ classical flags. For the outcome $x_1,\dots,x_j\in\{0,1\}^n$ let
\begin{equation}\label{eq:xentryOfPsi}
\begin{split}
\Psi_{x_1,\dots,x_j}&\coloneqq\bra{x_1,\dots,x_j}(T_1\otimes\dots\otimes T_k)\left(\sum_{i,i'=1}w_i\overline{w_{i'}}\ketbra{i}{i'}\otimes\vectorstate{\psi}^{\otimes n}\right)\ket{x_1,\dots,x_j}  \\
 &= \sum_{i,i'=1}^{t}w_i\overline{w_{i'}}\ketbra{i}{i'}\otimes\bigotimes_{m=1}^n\ket{\psi_{x_{1,m},x_{2,m},\dots,x_{k,m}}(z_i)}
 \bra{\psi_{x_{1,m},x_{2,m},\dots,x_{k,m}}(z_{i'})}\\
 &= \left(\bigoplus_{i}w_i\bigotimes_{m=1}^n\ket{\psi_{x_{1,m},x_{2,m},\dots,x_{k,m}}(z_i)}\right)
\left( \bigoplus_{i'}\overline{w_{i'}}\bigotimes_{m=1}^n\bra{\psi_{x_{1,m},x_{2,m},\dots,x_{k,m}}(z_{i'})}\right).
\end{split}
\end{equation}
Here the $2^k$ vectors $\psi_{x_{1,m},x_{2,m},\dots,x_{k,m}}(z_i)$ indexed by bit strings are obtained by applying one of $\frac{A_j(z_i)}{\norm{A_j(z_i)}}$ and $\sqrt{I-\frac{A_j(z_i)^*A_j(z_i)}{\norm{A_j(z_i)}^2}}$ at each of the $j$ factors depending on the respective bit being $0$ or $1$.

In the following we work with fixed outcome bit strings $x_1,\dots,x_j\in\{0,1\}^n$, such that they contain $\floor{Rn}$ times the all-$1$ flag $(1,\dots, 1)$, i.e., the number of $m$ indices such that $(x_{1,m},\dots,x_{j,m})=(1,\dots, 1)$ is $\floor{Rn}$.
We want to trace out the rest, but this can not be done trivially on the tensor sum, because these states are entangled. Let us consider one term at a time. By tracing out the unwanted copies we get the transformation
\begin{equation}\label{eq:onesummandtransform}
\begin{split}
\bigotimes_{m=1}^n\ket{\psi_{x_{1,m},x_{2,m},\dots,x_{k,m}}(z_i)} &\to \ket{\psi_{1\dots 1}(z_i)}^{\otimes \floor{Rn}}
    \norm{\bigotimes_{m= \floor{Rn}+1}^n\psi_{x_{1,m},x_{2,m},\dots,x_{k,m}}(z_i)}\\
    &=\ket{\psi_{1\dots 1}(z_i)}^{\otimes \floor{Rn}}
    \prod_{m= \floor{Rn}+1}^n\norm{\psi_{x_{1,m},x_{2,m},\dots,x_{k,m}}(z_i)}.
\end{split}
\end{equation} 
Here $\ket{\psi^{(i)}_{1\dots 1}}$ denotes the one-copy output state, which has a flag $1$ on each of its output bits. Without the loss of generality we ordered these all-$1$ flag bits to be the first $\floor{Rn}$ states.

To apply this to the tensor sum we use \cite[Proposition 2.14]{jensen2019asymptotic} to conclude that the transformation
\begin{multline}\label{eq:transformationoneoutput}
    \sum_{i=1}^{t}w_i\ket{i}\otimes\bigotimes_{m=1}^n\ket{\psi_{x_{1,m},x_{2,m},\dots,x_{k,m}}(z_i)}  \\  \to
    \sum_{i=1}^{t}w_i
    \left(
    \prod_{m= \floor{Rn}+1}^n\norm{\psi_{x_{1,m},x_{2,m},\dots,x_{k,m}}(z_i)}
    \right)
    \ket{i}\otimes\ket{\psi_{1\dots 1}(z_i)}^{\otimes \floor{Rn}}
\end{multline}
can also be performed with probability $1$. In the final step of the protocol we act with a rank-$1$ projection $\frac{1}{\sqrt{t}}\sum_{i=1}^t \ket{i}$ on the GHZ part.  
At the end we are left with the state
\begin{equation}\label{eq:finalvarphi}
    \sum_{i=1}^{t}c_i \norm{A_1(z_i)}^{ \floor{Rn}}\dots \norm{A_k(z_i)}^{ \floor{Rn}} \ket{\psi_{1\dots 1}(z_i)}^{\otimes \floor{Rn}}=
        \sum_{i=1}^{t}c_i \left(A_1(z_i)\otimes\dots\otimes A_k(z_i)\ket{\psi}\right)^{\otimes \floor{Rn}},
\end{equation}
where
\begin{equation}
    c_i=\frac{1}{\sqrt{t}}\frac{w_i\left(
    \prod_{m= \floor{Rn}+1}^n\norm{\psi_{x_{1,m},x_{2,m},\dots,x_{k,m}}(z_i)}
    \right)}{\norm{A_1(z_i)}^{ \floor{Rn}}\dots \norm{A_k(z_i)}^{ \floor{Rn}}}.
\end{equation}
Choosing $c_i$ and $z_i$ such that they admit \eqref{eq:czconditions} we proceed the same way as in \eqref{eq:procedureTophi}:
\begin{equation}
\begin{split}
\sum_{i=1}^{t}c_i \left(A_1(z_i)\otimes\dots\otimes A_k(z_i)\ket{\psi}\right)^{\otimes \floor{Rn}}
  &= \sum_{i=1}^t c_i \sum_{h=0}^{\floor{Rn}e} z_i^h\varphi_h^{\otimes \floor{Rn}}  \\
 &= \sum_{h=0}^{\floor{Rn}e}\left(\sum_{i=1}^t z_i^hc_i\right)\varphi_h^{\otimes \floor{Rn}}  \\
 &= \varphi^{\otimes \floor{Rn}}.
\end{split}
\end{equation}

\subsection{Optimizing the probabilities}

The key observation on optimizing the probability of the previously described protocol is that the problem formally coincides with the one we have solved in the $R=1$ case. The main difference is that here we have created $ \floor{Rn}$ copies instead of $n$.
The probability of the transformation given by the protocol is encoded in the norm of the input GHZ state, since $\psi$ and $\varphi$ were chosen to be unit vectors.

Until this point we have considered fixed $x_1,\dots,x_k\in\{0,1\}^n$ outcomes, but any permutation of these $n$ number of $k$-indices leads to the same probability.
There are in total $2^{k}$ number of different $k$-tuples of bits, and we have sampled this set $n$ times.
Let $n_{(b_1,\dots,b_k)}$ be the number of occurrences of the bit string $(b_1,\dots,b_k)$ in $(x_1,\dots,x_k)$.
The total number of permutations of these $k$-tuples is
\begin{equation*}
\binom{n}{n_{(0,\dots,0)}, n_{(0,\dots,0,1)},\dots, n_{(1,\dots,1)}}\approx 2^{nH(P_n)},
\end{equation*}
 where $\entropy(P_n)$ is the Shannon entropy of the empirical probability distribution 
 \begin{equation}
 P_n(b_1,\dots,b_k)\coloneqq\frac{n_{(b_1,\dots,b_k)}}{n}
 \end{equation}
 defined by:
\begin{equation}
\entropy(P_n) = -\sum_{(b_1,\dots,b_k) \in \{0,1\}^k} P_n(b_1,\dots,b_k) \log P_n(b_1,\dots,b_k),
\end{equation}
where we used base-$2$ logarithm. Here $\approx$ denotes the equality up to a polynomial factor.
By accepting any of these permutations, the individual probabilities sum up and we gain this factor in the total probability:
\begin{equation}\label{eq:psmallR}
\begin{split}
    p(R)
    &\gtrapprox 2^{nH(P_n)}\left(\sum_{i=1}^t \lvert c_i\rvert^2\frac{\norm{A_1(z_i)}^{2\floor{Rn}}\dots \norm{A_k(z_i)}^{2 \floor{Rn}}}{
    \prod_{m= \floor{Rn}+1}^n\norm{\psi_{x_{1,m},x_{2,m},\dots,x_{k,m}}(z_i)}^2
    }\right)^{-1} \cdot\frac{1}{t}\\
    &\gtrapprox 2^{nH(P_n)}\left(\sum_{i=1}^t \lvert c_i\rvert^2\frac{\norm{A_1(z_i)}^{2\floor{Rn}}\dots \norm{A_k(z_i)}^{2 \floor{Rn}}}{
    \prod_{(b_1,\dots,b_k)\in\{0,1\}^{k}\setminus (1,\dots, 1)}
    \norm{\psi_{b_1,\dots, b_k}(z_i)}^{2P_n(b_1,\dots,b_k)n}
    }\right)^{-1},
\end{split}
\end{equation}
where $\gtrapprox$ denotes $\ge$ up to a polynomial factor.

\begin{proposition}
    The success probability in \eqref{eq:psmallR} of the protocol described in \cref{sec:errorexptrfratelessone} results in the error exponent
\begin{equation}\label{eq:optrRfinal}
\begin{split}
r_\textnormal{opt}(R)&= -h(R) +
\inf_{P_{\textnormal{cond}}\in\mathcal{P}(\{0,1\}^k\setminus (1,\dots,1))}\inf_{\sigma\in\mathcal{P}(\complexes)}\sup_{z\in \supp \sigma} \bigg[
        2R\sum_{j=1}^k\log \norm{A_j(z)}
        \\&\quad\,+2Re\int_\complexes
        \log\frac{\lvert t\rvert}{\lvert z-t\rvert}\ed\sigma(t)
        +(1-R) \relativeentropy{P_{\textnormal{cond}}}{\{\norm{\psi_{b_1,\dots, b_k}(z
        )}^2\}_{b_1,\dots, b_k}} \bigg]
        .
\end{split}
\end{equation}
\end{proposition}
\begin{proof}
 Let
\begin{equation}\label{eq:defAP}
\begin{split}
    A_{P}(z)&\coloneqq \frac{\prod_{j=1}^k\norm{A_j(z)}}{
    \prod_{(b_1,\dots,b_k)\in\{0,1\}^{k}\setminus (1,\dots, 1)}\norm{\psi_{b_1,\dots, b_k}(z)}^{\frac{P(b_1,\dots, b_k)}{R}}
    },
\end{split}
\end{equation}
where $P$ is a probability distribution over the $k$-bit strings such that $P(1,\dots, 1)=R$.
For a fixed $P$, we choose a series of empirical probability distributions $P_n$ such that they converge to $P$.
Then for any $0<\epsilon$ there exist $n$ large enough so that
\begin{equation}
    \prod_{(b_1,\dots,b_k)\in\{0,1\}^{k}\setminus (1,\dots, 1)}
    \norm{\psi_{b_1,\dots, b_k}(z_i)}^{2P_n(b_1,\dots,b_k)}\ge
    2^{ -\epsilon}
    \prod_{(b_1,\dots,b_k)\in\{0,1\}^{k}\setminus (1,\dots, 1)}\norm{\psi_{b_1,\dots, b_k}(z_i)}^{2P(b_1,\dots, b_k)}.
\end{equation}
Using this we further modify the right hand side of \eqref{eq:psmallR} to get
\begin{equation}
\begin{split}
    p 
    &\ge  
    2^{nH(P_n)}
    \left(\sum_{i=1}^t \lvert c_i\rvert^2\frac{\norm{A_1(z_i)}^{2\floor{Rn}}\dots \norm{A_k(z_i)}^{2\floor{Rn}}}{
    \left(
    \prod_{(b_1,\dots,b_k)\in\{0,1\}^{k}\setminus (1,\dots, 1)}\norm{\psi_{b_1,\dots, b_k}(z_i)}^{P(b_1,\dots, b_k)}
    \right)^{2n}\cdot 2^{ -2n\epsilon}
    }\right)^{-1} \\
    &\ge  
    2^{nH(P_n)}
    \left(\sum_{i=1}^t \lvert c_i\rvert^2\frac{\norm{A_1(z_i)}^{2Rn}\dots \norm{A_k(z_i)}^{2Rn}}{
    \left(
    \prod_{(b_1,\dots,b_k)\in\{0,1\}^{k}\setminus (1,\dots, 1)}\norm{\psi_{b_1,\dots, b_k}(z_i)}^{\frac{P(b_1,\dots, b_k)}{R}}
    \right)^{2Rn}
    }\right)^{-1} \cdot 2^{ -2n\epsilon} \\
    &\ge
    2^{nH(P_n)}\left(\sum_{i=1}^t \lvert c_i\rvert^2 A_{P}(z_i)^{2Rn}\right)^{-1}\cdot 2^{ -2n\epsilon}.
\end{split}
\end{equation}

Note that we did not lose probability in the first order of the exponent by the previous bounds. 
In the asymptotic limit the first term leads to $H(P)$ and $2^{ -2n\epsilon}$ leads to $2\epsilon$ in the error exponent. In the latter $\epsilon>0$ is arbitrary, so this term can be omitted. 

One can see that the second term takes the same form as \eqref{eq:problowerbound}.
To optimize this expression in the complex numbers $z_i$ we can directly use the results of the previous sections. \cref{cor:finalcountable} deals with this optimization, but we need to ensure that its conditions are satisfied. By assuming that the Laurent polynomial $A_1(z)\otimes\dots\otimes A_k(z)$ contains positive and negative powers and it is nowhere zero (see \cref{rem:degenInterestingcases}), the conditions of \cref{cor:finalcountable} are met for $\prod_{j=1}^k\norm{A_j(z)}$. Because the denominator of \eqref{eq:defAP} is bounded from above these conditions also hold for $A_P$. 

We write the optimal error exponent achievable by this protocol
\begin{equation}\label{eq:optimalrR}
\begin{split}
r(P)&= -H(P) + 2R \inf_{\sigma\in\mathcal{P}(\complexes)}\sup_{z\in \supp \sigma} \left[
        \log A_{P}(z)+
        e\int_\complexes
        \log\frac{\lvert t\rvert}{\lvert z-t\rvert}\ed\sigma(t)\right]\\&=
        -H(P) + 2R \inf_{\sigma\in\mathcal{P}(\complexes)}\sup_{z\in \supp \sigma} \bigg[
        \sum_{j=1}^k\log \norm{A_j(z)}\\&\quad\,- \sum_{(b_1,\dots,b_k)\in\{0,1\}^{k}\setminus (1,\dots, 1)}\frac{P(b_1,\dots, b_k)}{R}\log \norm{\psi_{b_1,\dots, b_k}(z)} 
        +
        e\int_\complexes
        \log\frac{\lvert t\rvert}{\lvert z-t\rvert}\ed\sigma(t)\bigg]
        .
\end{split}
\end{equation}

The entropy of the probability distribution $P$ can be split by using the chain rule $H(P)=h(R)+(1-R)H(P_{\textnormal{cond}})$ where the first term is the binary entropy of $R$ and the second is the entropy of the conditional probability distribution on the bit strings excluding the all-$1$ $k$-tuples, i.e., $P_{\textnormal{cond}}(b_1,\dots, b_k)=P(b_1,\dots, b_k)/(1-R)$ if $(b_1,\dots, b_k)\neq (1,\dots, 1)$ and $0$ otherwise. Then
\begin{equation}\label{eq:rPsplit}
\begin{split}
r(P)&=  -h(R)-(1-R)H(P_{\textnormal{cond}}) + 
\inf_{\sigma\in\mathcal{P}(\complexes)}\sup_{z\in \supp \sigma} \bigg[
        2R\sum_{j=1}^k\log \norm{A_j(z)}
        \\&\quad\,+
        2Re\int_\complexes
        \log\frac{\lvert t\rvert}{\lvert z-t\rvert}\ed\sigma(t)
        -(1-R) \sum_{(b_1,\dots,b_k)\in\{0,1\}^{k}}2P_{\textnormal{cond}}(b_1,\dots, b_k)\log \norm{\psi_{b_1,\dots, b_k}(z
        )} \bigg]
    \\&=  -h(R) +
    \inf_{\sigma\in\mathcal{P}(\complexes)}\sup_{z\in \supp \sigma} \bigg[
        2R\sum_{j=1}^k\log \norm{A_j(z)}+
        2Re\int_\complexes
        \log\frac{\lvert t\rvert}{\lvert z-t\rvert}\ed\sigma(t)
        \\&\quad\,-(1-R) \sum_{(b_1,\dots,b_k)\in\{0,1\}^{k}}2P_{\textnormal{cond}}(b_1,\dots, b_k)\log \norm{\psi_{b_1,\dots, b_k}(z
        )}+H(P_{\textnormal{cond}})\bigg] 
       \\&=  -h(R) +
\inf_{\sigma\in\mathcal{P}(\complexes)}\sup_{z\in \supp \sigma} \bigg[
        2R\sum_{j=1}^k\log \norm{A_j(z)}+
        2Re\int_\complexes
        \log\frac{\lvert t\rvert}{\lvert z-t\rvert}\ed\sigma(t)
        \\&\quad\,+(1-R) \relativeentropy{P_{\textnormal{cond}}}{\{\norm{\psi_{b_1,\dots, b_k}(z
        )}^2\}_{b_1,\dots, b_k}} \bigg]
        ,
\end{split}
\end{equation}
where $\relativeentropy{P}{Q}$ is the Kullback--Leibler divergence between the finitely supported probability distribution $P$ and the finitely supported non-negative function $Q$:
\begin{equation}
\relativeentropy{P}{Q} = \sum_{x \in \supp P} P(x) \log \frac{P(x)}{Q(x)}.
\end{equation}
By optimizing this expression in $P_{\textnormal{cond}}$ we conclude our statement.
\end{proof}

\begin{proposition}\label{prop:smallRsymmetric}
    Assume that $\prod_{j=1}^k\norm{A_j(z)}$ is centrally symmetric, i.e., $\prod_{j=1}^k\norm{A_j(z)}=\prod_{j=1}^k\norm{A_j(ze^{i\varphi})}$ for any $\varphi\in\reals$, and $\norm{\psi_{b_1,\dots, b_k}(z)}^2$ is also centrally symmetric for every $b_1,\dots, b_k\in\{0,1\}$.
    Then the optimal error exponent in \eqref{eq:optrRfinal} takes the form
\begin{equation}\label{eq:rRsymmetricform}
   r_\textnormal{opt}(R)=
    -h(R) +
\inf_{z\in \complexes} \bigg[
        2R\sum_{j=1}^k\log \norm{A_j(z)}
        -(1-R) \log\left(1-\norm{\psi_{1,\dots, 1}(z
        )}^2\right)\bigg]
\end{equation}
\end{proposition}
\begin{proof} 
By the assumptions $A_P(z)$ on the right hand side of \eqref{eq:optimalrR} is centrally symmetric and we can directly apply \cref{prop:symmetricA} to that equation to get
\begin{equation}
   r(P)= -H(P) + 2R \inf_{z\in\complexes}
        \log A_{P}(z)
\end{equation}
Then, similarly to \eqref{eq:rPsplit}, we use the chain rule for the Shannon entropy to write
\begin{equation}
        r(R,P_{\textnormal{cond}})=-h(R) + \inf_{z\in\complexes}\left[
        2R\sum_{j=1}^k\log \norm{A_j(z)}+(1-R) \relativeentropy{P_{\textnormal{cond}}}{(\norm{\psi_{b_1,\dots, b_k}(z
        )}^2)_{b_1,\dots, b_k}}
        \right]
        .
\end{equation}
The rest is done by taking the infimum over $P_{\textnormal{cond}}\in\mathcal{P}(\{0,1\}^k\setminus (1,\dots,1))$. Note that 
\begin{equation}
\relativeentropy{P_{\textnormal{cond}}}{(\norm{\psi_{b_1,\dots, b_k}(z
        )}^2)_{b_1,\dots, b_k}}\ge -\log\sum_{b_1,\dots, b_k}\norm{\psi_{b_1,\dots, b_k}(z
        )}^2 
\end{equation}
and choosing
\begin{equation}
    P_{\textnormal{cond}}(b_1,\dots, b_k)\coloneqq \frac{\norm{\psi_{b_1,\dots, b_k}(z
        )}^2}{\sum_{b_1,\dots, b_k}\norm{\psi_{b_1,\dots, b_k}(z
        )}^2 }
\end{equation}
we have
\begin{equation}
\relativeentropy{P_{\textnormal{cond}}}{(\norm{\psi_{b_1,\dots, b_k}(z
        )}^2)_{b_1,\dots, b_k}}= -\log\sum_{b_1,\dots, b_k}\norm{\psi_{b_1,\dots, b_k}(z
        )}^2.
\end{equation}
From the observation that 
\begin{equation}
\bra{\psi}I_1\otimes\dots\otimes\frac{A_j(z_i)^*A_j(z_i)}{\norm{A_j(z_i)}^2}\otimes\dots\otimes I_k\ket{\psi}+\bra{\psi}I_1\otimes\dots\otimes\left(I-\frac{A_j(z_i)^*A_j(z_i)}{\norm{A_j(z_i)}^2}\right)\otimes\dots\otimes I_k\ket{\psi}=
\norm{\psi}^2,
\end{equation}
meaning that 
\begin{equation}
    \sum_{b_1,\dots,b_k\in \{0,1\}}
    \norm{\psi_{b_1,b_2,\dots,b_k}(z_i)}^2=1,
\end{equation}
follows
\begin{equation}
     \log\sum_{b_1,\dots, b_k}\norm{\psi_{b_1,\dots, b_k}(z
        )}^2= \log\left(1-\norm{\psi_{1,\dots, 1}(z
        )}^2\right),
\end{equation}
which concludes our proof.
\end{proof}
\begin{remark}
    Observe that if $A_j(z)^*A_j(z)$ is centrally symmetric then it implies this symmetry on the norm $\norm{A_j(z)}$, and also on $\norm{\psi_{b_1,\dots, b_k}(z)}^2$ which is obtained by applying the operators $\frac{A_j(z)^*A_j(z)}{\norm{A_j(z)}^2}$ or $\sqrt{I-\frac{A_j(z)^*A_j(z)}{\norm{A_j(z)}^2}}$ for each $j\in [k]$ on the input state $\psi$ and taking its inner product with $\psi$. Therefore the central symmetry of $A_j(z)^*A_j(z)$ is a stronger, but more convenient condition than what is used in \cref{prop:smallRsymmetric}. Although the central symmetry of $A_j(z)^*A_j(z)$ may not hold for all degenerations with centrally symmetric norms, it holds for any combinatorial degeneration (see \cref{def:combinatorialDegen}).
\end{remark}

\begin{example}\label{ex:GHZtoWSmallR}
    We can repeat \cref{ex:GHZtoW} considering a transformation rate $0<R<1$. Using that $A_j(z)^*A_j(z)$ is centrally symmetric we can apply $\cref{prop:smallRsymmetric}$.
    \begin{equation}
    \begin{split}
        \norm{\psi_{1\dots 1}(\lvert z\rvert)}^2&=
        \frac{\norm{A_1(z)\otimes\dots\otimes A_k(z) \psi}^2}{\norm{A_1(z)\otimes\dots\otimes A_k(z)}^2}
        =
        \frac{1}{\norm{A_1(z)\otimes\dots\otimes A_k(z)}^2}
        \frac{3+\lvert z\rvert^{12}}{3}
        \\[10pt] &=
        \begin{cases}    
        \frac{3+\lvert z\rvert^{12}}{4\lvert z\rvert^{12}}=
        \frac{3}{4\lvert z\rvert^{12}}+\frac{1}{4}
        \quad\text{if}\quad\lvert z\rvert\ge 1\\
        \frac{3+\lvert z\rvert^{12}}{4\lvert z\rvert^{-6}}=
        \frac{3\lvert z\rvert^6}{4}+\frac{\lvert z\rvert^{18}}{4}
        \quad\text{if}\quad\lvert z\rvert< 1
        \end{cases}
        \end{split}
    \end{equation}
     The results of the numerical optimization of \eqref{eq:rRsymmetricform} are shown on \cref{fig:GHZtoWcombinat}.
\end{example}

\begin{figure}
\begin{center}
\begin{tikzpicture}
\begin{axis}[
  rateplotstyle,
  ymin=0,
  ymax=1.1,
  extra y ticks={0},
  xtick={0,0.05,0.1,0.15,0.2,0.25,0.3,0.35,0.415 },
  xticklabels={0,0.05,0.1,0.15,0.2,0.25,0.3,0.35,0.415},
  xmin=0,
  xmax=0.445,
  xlabel={$\quad r$},
  ylabel={$R$},
  samples=100,
]
\path[name path=TOP] (\pgfkeysvalueof{/pgfplots/xmin},\pgfkeysvalueof{/pgfplots/ymax})--(\pgfkeysvalueof{/pgfplots/xmax},\pgfkeysvalueof{/pgfplots/ymax});
\path[name path=BOTTOM] (\pgfkeysvalueof{/pgfplots/xmin},1)--(\pgfkeysvalueof{/pgfplots/xmax},1);
\addplot[plotline,name path=UBOUND] {1};
\addplot[notachievableregion, pattern=north east lines] fill between [of=BOTTOM and TOP];

\addplot[plotline, name path=INVERTEDFUNC] table [x=r, y=R] {\invertedfunctiondata};
\path[name path=AXIS] (axis cs:0,0) -- (axis cs:0.415,0);
\addplot[achievableregion] fill between [of=INVERTEDFUNC and AXIS];
\draw[thick, dashed] (axis cs:0.415, 0) -- (axis cs:0.415,1);
\draw[thick, dotted] (axis cs:0, 0) -- (axis cs:0.415,1);
\path[name path=RIGHTACHIEVABLE] (axis cs:0.415,0) -- (axis cs:0.445,0) -- (axis cs:0.445,1) -- (axis cs:0.415,1) -- cycle;
\addplot[achievableregion] fill between [of=RIGHTACHIEVABLE and AXIS];

\end{axis}
\end{tikzpicture}
\end{center}
\caption{The rate-error exponent plane for transforming \GHZ{} states into \W{} states. The curve shows the achievable $(R,r)$ pairs calculated by numerical optimization of \cref{prop:smallRsymmetric}. For $R=1$ we recover the error exponent $r\approx 0.415$ calculated in \cref{ex:GHZtoW}, then for larger error exponents the rate $R=1$ is also achievable. The shaded region below the curve is also achievable, but the lined region ($R>1$) is not achievable due to the equality of the local ranks.
The dotted line shows the trade-off curve of the time-sharing protocol utilizing the known protocol for $R=1$.   
The achievability of the white region is not determined by these bounds.}
\label{fig:GHZtoWcombinat}
\end{figure}
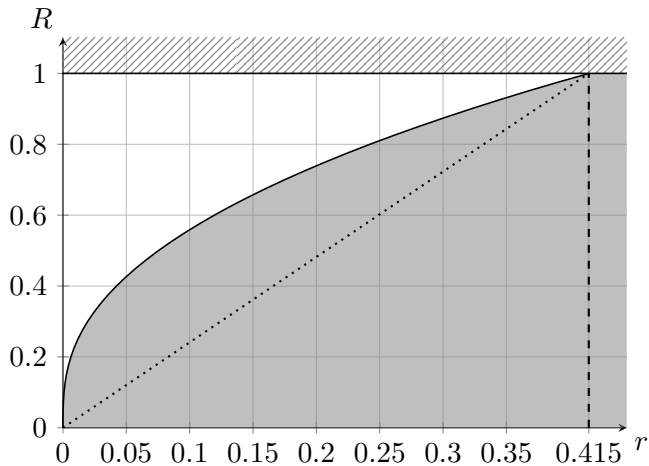

\begin{corollary}
    Assume the symmetry conditions of \cref{prop:smallRsymmetric}. Then the best achievable error exponent is
    \begin{equation}
    r_{\textnormal{opt}}=  -\sup_{z\in\complexes}\log
        \bigg(
        \exp\left(
        -2\sum_{j=1}^k\log \norm{A_j(z)}
        \right)
        + \exp\left(
        \vphantom{\sum_{j=1}^k}
        \log\left(1-\norm{\psi_{1,\dots, 1}(z
        )}^2\right)\right)  \bigg).
    \end{equation}
\end{corollary}
\begin{proof}
    We take the infimum of the right hand side of $\eqref{eq:rRsymmetricform}$ over $R$. Note that the entropy $h(R)$ can be moved inside the infimum over $z$, which can be interchanged with the infimum over $R$. Using
\begin{equation}
        \sup_{p\in [0,1]}\left[ h(p)+ px_1 + (1-p)x_2\right]
        =\log\big( 2^{x_1}+2^{x_2}\big)
\end{equation}
from \cite[Eq. (2.13)]{strassen1991degeneration} we get
    \begin{equation}
    \begin{split}
&\inf_{R\in [0,1]} \bigg[ 
        2R\sum_{j=1}^k\log \norm{A_j(z)}
        -(1-R) \log\left(1-\norm{\psi_{1,\dots, 1}(z
        )}^2\right)  -h(R)\bigg]\\
        =&-
        \sup_{R\in [0,1]} \bigg[ 
        -2R\sum_{j=1}^k\log \norm{A_j(z)}
        +(1-R) \log\left(1-\norm{\psi_{1,\dots, 1}(z
        )}^2\right)  +h(R)\bigg]\\
        =&-\log
        \bigg[
        \exp\left(
        -2\sum_{j=1}^k\log \norm{A_j(z)}
        \right)
        + \exp\left(
        \vphantom{\sum_{j=1}^k}
        \log\left(1-\norm{\psi_{1,\dots, 1}(z
        )}^2\right)\right)  \bigg]
    \end{split}
    \end{equation}
\end{proof}

\section{Conclusion}

We have shown that if a state $\varphi$ arises as a degeneration of another state $\psi$, then an asymptotic probabilistic transformation from $\psi$ to $\varphi$ is possible with a finite strong converse exponent, more precisely, with an exponent bounded by a single-letter formula given in \cref{thm:main}. In addition, we provide a bound on the trade-off relation between the transformation rate $R\in(0,1)$ and the error exponent, which is more favourable than the bound that follows from a convexity (time-sharing) argument.

However, our bound on the success probability remains exponentially small even for very small rates $R$, where one would expect that the probability of success in fact approaches $1$ exponentially fast. In contrast, if a single-copy LOCC transformation from $\psi$ to $\varphi$ is possible with probability $p$, then for any rate $R<p$ there exists a protocol for an asymptotic transformation that succeeds with probability $1-2^{-nd(R\|p)}$, where $d(q\|p)=p\log\frac{p}{q}+(1-p)\log\frac{1-p}{1-q}$. It would be most interesting to construct a protocol from a degeneration that, for some nonzero rate $R$ succeeds with probability $1-2^{-\Omega(n)}$. This would apply even in cases when there is a degeneration from $\psi$ to $\varphi$ (so that rate $1$ is achievable with asymptotic SLOCC), but there is no SLOCC transformation from $\psi^{\otimes n}$ to $\varphi^{\otimes n}$ for any $n$.

Concerning the other limiting case of large strong converse exponents, it is not even known if the optimal asymptotic SLOCC rate can be achieved with a success probability that goes to $0$ exponentially (and not faster). Denoting the largest achievable rate for a given strong converse exponent $r$ by $R^*(\psi\to\varphi,r)$, this amounts to asking if the function $R^*(\psi\to\varphi,r)$ is eventually constant. When the optimal rate is achieved by a finite-copy SLOCC transformation (restriction), then by running the protocol independently on copies (blocks) already gives and exponential lower bound, and our result shows that exponential behaviour also follows when the optimal rate arises from a finite-copy degeneration. However, this leaves open the possibility that the largest success probability at the optimal SLOCC rate decreases faster than any exponential in case the optimal rate can only be achieved by a sequence of degenerations. As we have seen in \cref{ex:graphtensors}, a finite error exponent at the optimal rate may still be possible even when no degeneration exists for any finite number of copies with the same rate.

\section*{Acknowledgement}

The project has been implemented with the support provided by the Ministry of Culture and Innovation of Hungary from the National Research, Development and Innovation Fund, financed under the FK~146643 funding scheme. Supported by the \'UNKP-23-5-BME-458 New National Excellence Program of the Ministry for Culture and Innovation from the source of the National Research, Development and Innovation Fund, and the J\'anos Bolyai Research Scholarship of the Hungarian Academy of Sciences, and the National Research, Development and Innovation Office within the Quantum Information National Laboratory of Hungary (Grant No.~2022-2.1.1-NL-2022-00004).

\bibliography{references}{}

\newpage
\appendix

\section{Transfinite diameter and logarithmic capacity}\label{subsec:transfiniteNpot}

The transfinite diameter measures the ``size'' of a compact set in the complex plane or more generally in Euclidean space. For a compact set $K \subset \mathbb{C}$, the transfinite diameter $d(K)$ is defined as
\begin{equation}
d(K) = \lim_{n \to \infty} \left( \sup_{z_1, z_2, \ldots, z_n \in K} \prod_{1 \leq i < j \leq n} \lvert z_i - z_j\rvert \right)^{\frac{2}{n(n-1)}}.
\end{equation}
This limit exists and provides a measure of the distribution of points in $K$. This concept can be generalized to the weighted transfinite diameter
\begin{equation}
d_w(K) = \lim_{n \to \infty} \left( \sup_{z_1, z_2, \ldots, z_n \in K} \prod_{1 \leq i < j \leq n} \bigg(w(z_i)w(z_j)\lvert z_i - z_j\rvert\bigg)^{\frac{2}{n(n-1)}} \right),
\end{equation}
where $w: \complexes\to\reals$ is the weight function.

There is an equivalent characterisation of the transfinite diameter, namely the logarithmic capacity. For a compact set $K \subset \mathbb{C}$, let
\begin{equation}
V(K)\coloneqq \inf_{\mu\in\mathcal{P}(K)} I(\mu) = \iint \log \frac{1}{\lvert z - t\rvert} \ed\mu(z) \ed\mu(t).
\end{equation}
The capacity $c(K)$ is then given by:
\begin{equation}
c(K) = e^{-I(\mu_K)}.
\end{equation}
It has been shown that $c(K)=d(K)$ for a compact set $K$ (\cite{ransford1995potential}).
While the potential integral describes the potential caused by the charges, the weight functions $w(z)$ in the weighted logarithmic capacity 
\begin{equation}\label{eq:weightedCap}
    I_w(\mu) \coloneqq \iint \log \frac{1}{w(z)w(t)\lvert z - t\rvert}\ed\mu(z)\ed\mu(t)
\end{equation}
account for the potential caused by an external field. 
Similarly to the uniformly weighted case we have $c_w(K)\coloneqq e^{-I_w(\mu_K)}=d_w(K)$ (\cite{saff2013logarithmic}).

For a measure $\mu$ on a compact set $K \subset \complexes$, the potential $U^\mu$ at a point $z \in \complexes$ is given by:
\begin{equation}
U^\mu(z) = \int_K \log \frac{1}{\lvert z - t\rvert}\ed\mu(t).
\end{equation}
This integral describes how the measure $\mu$ influences the potential at $x$.

A function $u: \complexes \to \mathbb{R}$ is harmonic if it satisfies Laplace's equation:
\begin{equation}
\Delta u = 0,
\end{equation}
where $\Delta$ is the Laplacian operator. A key property of a harmonic function is that it attains its maximum and minimum on the boundary of its domain. In the following we show that the potential function is harmonic outside of $\supp\mu$, from which follows that it can not have local extrema outside of $\supp\mu$ except if it is constant there.

\begin{lemma}\label{lem:harmonic}
    Let $t$ and $z$ be complex numbers. Then $U^\mu(z)=-\int\log{\lvert z-t\rvert}\ed\sigma(t)$ is harmonic outside the support of $\sigma$, i.e., $\Delta U^\mu(z) = 0$ for any $z\in \complexes\setminus\supp\sigma$.
\end{lemma}
\begin{proof}
Let $z = x + iy$ and $t = a + ib$, where $x, y, a, b$ are real numbers and $z\neq t$. First we aim to show that the function $u(z) = \log \lvert z - t\rvert$ is harmonic with respect to $z$. 
The function $u(z)$ can be written as:
\begin{equation}
u(z)=u(x, y) = \log \lvert z - t\rvert = \log \sqrt{(x - a)^2 + (y - b)^2}.
\end{equation}
To prove that $u(x, y)$ is harmonic, we need to show that it satisfies the Laplace equation
\begin{equation}
\frac{\partial^2 u}{\partial x^2} + \frac{\partial^2 u}{\partial y^2} = 0.
\end{equation}
First, we compute the partial derivatives of $u(x, y)$
\begin{equation}
\frac{\partial u}{\partial x} = \frac{1}{2} \cdot \frac{2(x - a)}{(x - a)^2 + (y - b)^2} = \frac{x - a}{(x - a)^2 + (y - b)^2},
\end{equation}
\begin{equation}
\frac{\partial u}{\partial y} = \frac{1}{2} \cdot \frac{2(y - b)}{(x - a)^2 + (y - b)^2} = \frac{y - b}{(x - a)^2 + (y - b)^2}.
\end{equation}
Next, we compute the second-order partial derivatives using the quotient rule:
\begin{equation}
\frac{\partial^2 u}{\partial x^2} = \frac{(x - a)^2 + (y - b)^2 - (x - a) \cdot 2(x - a)}{\left( (x - a)^2 + (y - b)^2 \right)^2} = \frac{(y - b)^2 - (x - a)^2}{\left( (x - a)^2 + (y - b)^2 \right)^2},
\end{equation}
\begin{equation}
\frac{\partial^2 u}{\partial y^2} = \frac{(x - a)^2 + (y - b)^2 - (y - b) \cdot 2(y - b)}{\left( (x - a)^2 + (y - b)^2 \right)^2} = \frac{(x - a)^2 - (y - b)^2}{\left( (x - a)^2 + (y - b)^2 \right)^2}.
\end{equation}
Adding these second-order partial derivatives together, we get
\begin{equation}
\frac{\partial^2 u}{\partial x^2} + \frac{\partial^2 u}{\partial y^2} = \frac{(y - b)^2 - (x - a)^2}{\left( (x - a)^2 + (y - b)^2 \right)^2} + \frac{(x - a)^2 - (y - b)^2}{\left( (x - a)^2 + (y - b)^2 \right)^2} = 0.
\end{equation}
Hence, we have shown that
\begin{equation}
\frac{\partial^2 u}{\partial x^2} + \frac{\partial^2 u}{\partial y^2} = 0
\end{equation}
for any $z\neq t$. By the linearity of the integration then it also follows that $\int\log{\lvert z-t\rvert}\ed\sigma(t)$ is harmonic for $z\not\in\supp\sigma$.
\end{proof}

\end{document}